\documentclass[journal]{IEEEtran}
\ifCLASSINFOpdf
\else
\fi
\usepackage{amsmath,amssymb,amsthm,mathrsfs,xfrac,cite,subfigure}
\usepackage{newtxmath}
\usepackage{xcolor}
\usepackage{hyperref}

\usepackage{algorithmic}

\newcommand{\va}{{\boldsymbol a}}

\newcommand{\vX}{{\boldsymbol X}}
\newcommand{\vU}{{\boldsymbol U}}
\newcommand{\vV}{{\boldsymbol V}}
\newcommand{\vj}{{\boldsymbol j}}
\newcommand{\vJ}{{\boldsymbol J}}
\newcommand{\vz}{{\boldsymbol z}}

\newcommand{\vp}{{\boldsymbol p}}
\newcommand{\vq}{{\boldsymbol q}}
\newcommand{\bv}{{\boldsymbol v}}   
\newcommand{\vA}{{\boldsymbol A}}
\newcommand{\vR}{{\boldsymbol R}}
\newcommand{\vB}{{\boldsymbol B}}
\newcommand{\vg}{{\boldsymbol g}}

\newcommand{\vH}{{\boldsymbol H}}
\newcommand{\vM}{{\boldsymbol M}}
\newcommand{\vG}{{\boldsymbol G}}
\newcommand{\vS}{{\boldsymbol S}}

\newcommand{\vI}{{\boldsymbol I}}
\newcommand{\vx}{{\boldsymbol x}}
\newcommand{\vr}{{\boldsymbol r}}
\newcommand{\vs}{{\boldsymbol s}}
\newcommand{\vu}{{\boldsymbol u}}

\newcommand{\ve}{{\boldsymbol e}}
\newcommand{\vh}{{\boldsymbol h}}

\newcommand{\vzero}{{\boldsymbol 0}}

\newtheorem{lemma}{Lemma}
\newtheorem{theorem}{Theorem}

\usepackage{booktabs}
\usepackage{multirow}

\begin{document}
%
\title{Solving Complex Quadratic Systems \\with Full-Rank Random Matrices}
%
%
%

\author{Shuai~Huang,
        Sidharth~Gupta,
        and~Ivan~Dokmani\'c,~\IEEEmembership{Member,~IEEE}
        \thanks{\copyright 2020 IEEE. Personal use of this material is permitted. Permission from IEEE must be obtained for all other uses, in any current or future media, including reprinting/republishing this material for advertising or promotional purposes, creating new collective  works,  for  resale  or  redistribution  to  servers  or  lists,  or  reuse  of  any  copyrighted  component  of  this  work  in  other works.}\thanks{This work is supported by National Science Foundation under Grant CIF-1817577. The authors are with the Coordinated Science Laboratory, University of Illinois at Urbana-Champaign, Urbana, IL 61801 (e-mail: shuai.huang@emory.edu, gupta67@illinois.edu, dokmanic@illinois.edu).}
}
%
%


\maketitle

\begin{abstract}
We tackle the problem of recovering a complex signal $\vx\in\mathbb{C}^n$ from quadratic measurements of the form $y_i=\vx^*\vA_i\vx$, where $\vA_i$ is a full-rank, complex random measurement matrix whose entries are generated from a rotation-invariant sub-Gaussian distribution. We formulate it as the minimization of a nonconvex loss. This problem is related to the well understood phase retrieval problem where the measurement matrix is a rank-1 positive semidefinite matrix. Here we study the general full-rank case which models a number of key applications such as molecular geometry recovery from distance distributions and compound measurements in phaseless diffractive imaging. Most prior works either address the rank-1 case or focus on real measurements. The several papers that address the full-rank complex case adopt the computationally-demanding semidefinite relaxation approach. In this paper we prove that the general class of problems with rotation-invariant sub-Gaussian measurement models can be efficiently solved with high probability via the standard framework comprising a spectral initialization followed by iterative Wirtinger flow updates on a nonconvex loss. Numerical experiments on simulated data corroborate our theoretical analysis.
\end{abstract}


%
\IEEEpeerreviewmaketitle

\section{Introduction}
\label{sec:intro}
\IEEEPARstart{S}{ystems} of quadratic equations model many problems in applied science, including phase retrieval \cite{Fienup:82,Millane:90,Shechtman2015:PRRev,Jaganathan2015PhaseRA}, the unlabeled distance geometry problem (uDGP) \cite{DUXBURY2016117,Huang2018uDGP}, the turnpike and beltway problems \cite{Dakic2000TP,Lemke2003,Huang2018uDGP}, unknown view tomography \cite{Basu2000UVT,Zehni2019UVT,Zehni2020UVT}, blind channel estimation \cite{Ahmed2014BCE,Ranieri2013PRSparse}, power flow analysis and power system state estimation \cite{Wang2017power}. Phase retrieval, in particular, has motivated considerable recent research on quadratic equations. The phaseless measurements are given as $y_i = |\va_i^*\vx|^2 = \vx^* \va_i \va_i^* \vx$, with the measurement matrices $\va_i\va_i^*$ being rank-1 positive semidefinite matrices. In this paper we study a different measurement model with full-rank measurement matrices. Such measurements arise in a number of the aforementioned applications.

In combinatorial optimization problems such as the uDGP\cite{DUXBURY2016117,Huang2018uDGP} and the nanostructure problem \cite{Juhas:2006,Billinge2010_nanostructure}, the goal is to reconstruct the relative locations of a set of points from their unlabeled pairwise distances. The distribution $y_i$ of the distance $i$ can be formulated as a quadratic form with respect to the point density \cite{Huang2018uDGP}, with the measurement matrix being high-rank. Unknown view tomography aims to reconstruct a 3D density map from a collection of 2D projection images with unknown view angles. When the view angles are assumed to be uniformly distributed on the unit sphere, rotation invariant features can be estimated from 2D projection images and later used for reconstruction \cite{Zehni2020UVT}. Many of these features can be written as quadratic forms with respect to high-rank measurement matrices.

A variety of combinatorial graph problems can be formulated as quadratic problems with high-rank measurement matrices \cite{helmberg1998solving}, including the problem of finding the minimum energy spin configuration of atoms located on a grid in quantum physics \cite{poljak1995maximum}, and the problem of minimizing the number of connections between layers of a circuit board in very-large-scale-integrated (VLSI) circuit design \cite{barahona1988application}. Furthermore, in an electric transmission network consisting of nodes (buses) and edges (transmission lines), power flow analysis tries to compute the complex voltage at all nodes given measured or specified system variables at selected nodes and edges. This task is then cast as solving a system of quadratic equations with respect to the complex voltages \cite{Wang2017power} with the measurement matrix being sparse and with rank greater than one. For noisy measurements the task is known as power system state estimation.

These problems can be modeled as systems of quadratic equations where the measurement matrices are not necessarily rank-1 or real. Recovery of a signal from its complex quadratic measurements is naturally formulated as a nonconvex optimization problem, where solving for the globally optimal solution is in general intractable. Recent works on nonconvex quadratic problems such as phase retrieval \cite{Netrapalli2015:RPAM,Candes2015PhaseRV}, phase synchronization \cite{NoncvxPhaseSync2016,PhaseSync2017}, and low-rank matrix recovery \cite{Chen2015FastLE} have shown that a globally optimal solution can be recovered from sufficient measurements with high probability when iid Gaussian measurement vectors or matrices are used. Light transport in random media can be modeled by an iid complex Gaussian matrix, and dedicated hardwares like optical processing units (OPU) have been used to produce rapid (20kHz) random projections of high-dimensional data in the million range \cite{OPU:2019}. In the calibration experiment where the calibration signals are controlled and known, the phases of the calibration measurements can be recovered and used to estimate the complex Gaussian measurement vectors \cite{OPU:2020}. This application bridges the gap between well-understood random measurement theories and real applications. It further motivates the question of what other types of measurement vectors or matrices also enjoy such favorable properties. In particular, extending the Gaussian measurement model to the sub-Gaussian case has recently attracted considerable interests \cite{Chen:Quad2015,Krahmer_2018,gao2019phase,Krahmer:2019}. Deriving the theory for other ensembles such as Bernoulli \cite{Krahmer_2018,Krahmer:2019} has the potential to explain a wider range of applications.

Building upon our earlier work on random Gaussian measurements in \cite{QE:Gaussian:2019}, we show that our results hold for a slightly larger class of measurement matrices which we term rotation-invariant sub-Gaussian matrices. In this case the signal of interest $\vx\in\mathbb{C}^n$, the full-rank sub-Gaussian measurement matrices $\vA_i\in\mathbb{C}^{n\times n}$, the measurement $y_i\in\mathbb{C}$ are all in the complex domain and we have
\begin{align}
\label{eq:quad_eqs}
y_i=\vx^*\vA_i\vx, \quad i=1,\ldots,m\,.
\end{align}
We propose to recover a globally optimal solution via the standard framework comprising a spectral initialization and iterative Wirtinger flow (WF) updates. We prove that when the number of measurements $m$ exceeds the signal length $n$ by some sufficiently large constant $C$, i.e. $m>Cn$, the signal $\vx$ can be recovered up to a global phase shift with high probability. 

\subsection{Prior art}
Similar quadratic equation problems have been studied in other contexts. Cand\`{e}s et al. \cite{Candes2015PhaseRV} cast the phase retrieval problem as a system of structured quadratic equations and solved it via WF with a linear convergence rate. As this is a non-convex problem, they used a suitably constructed spectral initializer, $\vz^{(0)}$, for the Gaussian measurement model. Spectral initialization for phase retrieval was originally proposed in \cite{Netrapalli2015:RPAM}. It produces $\vz^{(0)}$ which is close to a globally optimal solution with high probability when sufficient measurements are available. The works of \cite{Chen2015Trunc,Wang2018Trunc} subsequently showed that adapting the loss and truncating the measurements adaptively in the initialization and gradient stages could lead to improved performance. Although the spectral initialization was originally developed for random Gaussian measurements, it can also be adapted to work with other types of measurements \cite{Yurtsever:2017,Huang2018uDGP}.

Additionally, as shown in the proofs of the WF approach \cite{Candes2015PhaseRV}, in phase retrieval some of the entries of the measurement matrix are correlated. This makes it impossible to use some of the well-established random matrix theory \cite{vershynin_2012}. In our measurement model, the matrix entries are pairwise uncorrelated and have zero mean, which leads to a different and much simplified proof to establish the convergence guarantees.

Lu and Li \cite{Lu18Spectral} studied generalizations of spectral initialization in the real case and focused on the asymptotic behavior of the initializer with respect to the sampling ratio $m/n$ in the high-dimensional limit. Moving beyond the Gaussian measurement model, Ghods et al. \cite{Ghods2018LinearSE} proposed a linear spectral estimator for general nonlinear measurement systems. The works of Wang and Xu \cite{WANG2017,XU2018497} addressed a generalized phase retrieval problem where $\vA_i$ is a Hermitian matrix. They used algebraic methods \cite{CONCA2015346} to find the number of measurements needed for a successful recovery. Here we build upon these results and show that both the initialization and convergence proofs can be derived using the Bernstein-type inequalities for the full-rank rotation-invariant sub-Gaussian measurement model in \eqref{eq:quad_eqs}.

Solving systems of quadratic equations is closely related to low-rank matrix recovery---it is equivalent to recovering a rank-1 positive semidefinite (PSD) matrix $\vX=\vx\vx^*$ with $y_i = \langle \vA_i, \vX \rangle$ \cite{CANDES2015CDF, candes2015phase}. The nonconvex low-rank constraint on $\vX$ can be relaxed to the convex minimum nuclear norm constraint. Alternatively, Carlsson and Gerosa \cite{Carlsson_2019} proposed to directly search for a PSD matrix with known rank $k$ and used it to perform phase retrieval from Fourier measurements. The works of \cite{Recht2010NNM,ELDAR2012309,Kabanava2016NSP,kueng2017low} focused on establishing sufficient conditions that warrant such a semidefinite convex relaxation. However, recovering the relaxed $\vX$ is computationally demanding even for moderate-size problems. To address this issue, Yurtsever et al. \cite{yurtsever2017sketchy} modified the conditional gradient method by performing a small random sketch of $\vX$ that can be used to recover $\vX$ later. On the other hand, the works in \cite{Chen2015FastLE,NIPS2015_5733,NIPS2015_5830,Tu2016Procrustes} approached the problem directly by factoring the unknown matrix as $\vM=\vU\vV^T$, where $\vU\in\mathbb{R}^{n_1\times k}, \vV\in\mathbb{R}^{n_2\times k}$ and $\vM\in\mathbb{R}^{n_1\times n_2}$ are all real matrices, and searching for $\vU$ and $\vV$ instead. When $\vM$ is positive semidefinite and the measurement operator satisfies the restricted isometry property \cite{RIP_LR:2010}, Zheng and Lafferty focused on the Gaussian measurement model and proved that $\vM$ could be recovered with guarantees \cite{NIPS2015_5830}. Tu et al. improved upon \cite{NIPS2015_5830} and proposed the Procrustes Flow approach for the more general case where $n_1\neq n_2$ \cite{Tu2016Procrustes}. Our contribution to this line of work lies not just in going from the real case to the complex case in the rotation-invariant sub-Gaussian model (which requires technical interventions at every step), but also in establishing new proofs based on (and including) Lemma \ref{lemma:spectral_norm_concen_all} that offer a simpler way to obtain the recovery guarantees.

Furthermore, recent works on low-rank matrix recovery showed that the nonconvex problem enjoys a globally optimal landscape \cite{Srinadh:2016:GlobalOpt,Park:2017,Ge:2017:Geometric,Li:2018:NonconvexGeo}. For low-rank square, PSD matrix recovery, Bhojanapalli et al. \cite{Srinadh:2016:GlobalOpt} proved that there are no spurious local minima when using incoherent linear measurements defined by the RIP condition. In this case, a global convergence guarantee can be obtained for stochastic gradient descent from random initialization. Park et al. \cite{Park:2017} later generalized it to the non-square case. Via a geometric analysis on the Hessian of the objective function, Ge et al. \cite{Ge:2017:Geometric} proposed a way to find directions in which local minima can be ``improved''. Recent developments on using various nonconvex optimization approaches to solve low-rank matrix factorization can be found in the review paper by Chi et al. \cite{NonconvexLR}.

In this work we complement these results by studying the complex signal recovery problem: we aim to recover a complex signal $\vx\in\mathbb{C}^n$ from its complex quadratic measurements up to a global phase shift. 

\subsection{Paper outline}
The paper is organized as follows. In Section \ref{sec:spectral_init}, we extend the derivations from \cite{Netrapalli2015:RPAM} to the rotation-invariant sub-Gaussian measurement model. We show that the spectral initialization concentrates around a global optimum with high probability and compute the associated concentration bounds. In Section \ref{sec:convergence_analysis}, we analyze the regularity condition and derive new results for the rotation-invariant sub-Gaussian measurement model. The two results are then combined to give the main theorem of this paper. Computational experiments are presented in Section \ref{sec:experiments}. The proofs of the lemmas are given in the Appendix.

\section{Problem formulation}
\subsection{Rotation-invariant sub-Gaussian measurement model}
For convenience, let $\vr_i\in\mathbb{R}^{2n^2}$ denote the real and imaginary coefficients of the entries of $\vA_i$. The $m$ coefficient vectors $\vr_i$ for $1 \leq i \leq m$ are independent and identically distributed following a multivariate rotation-invariant sub-Gaussian distribution \cite{Bryc1995,vershynin_2012}.
\begin{itemize}
    \item The distribution of $\vr_i$ does not change under unitary transforms \cite{Bryc1995}. It follows that the probability density function $p\left(\vr_i\right)$ depends only on the norm $\|\vr_i\|_2$. Without loss of generality, we also assume $\mathbb{E}[r_{ik}^2]=1,\ \forall\ k=1,\cdots,2n^2$.
    \item The coefficient vector $\vr_i$ is a sub-Gaussian random vector \cite[Definition 5.22]{vershynin_2012} such that the one-dimensional marginals $\vg^T\vr_i$ are sub-Gaussian random variables for all $\vg \in \mathbb{R}^{2n^2}$. In particular, every single entry $r_{ik}$ is also sub-Gaussian by definition.
\end{itemize}
We shall refer to the model given by \eqref{eq:quad_eqs} as the rotation-invariant sub-Gaussian measurement model, which is a generalization of the Gaussian measurement model from our earlier work \cite{QE:Gaussian:2019}. Some examples of matrices in this model are given in Section \ref{sec:rot_exp_family}.

Instead of solving the system of quadratic equations in \eqref{eq:quad_eqs} directly, we formulate it as a minimization problem. Namely, we minimize the following loss function $f(\vz)$ to obtain the recovered signal $\vz$:
\begin{align}
\label{eq:first_obj}
\begin{split}
f(\vz) = \frac{1}{m}\textstyle\sum\limits_{i=1}^m \left|\vz^*\vA_i\vz-y_i\right|^2\,.
\end{split}
\end{align}
Clearly, for any solution $\vz_0$ to \eqref{eq:quad_eqs} we have $f(\vz_0) = 0$. Although we do not prove it here, one can expect that with sufficiently many ``generic'' measurements the map from $\vz$ to $[y_1, \ldots, y_m]^T$ is injective up to a global phase. Such results have been rigorously proven for the phase retrieval problem \cite{balan2006signal, conca2015algebraic}.

Suppose we are given a ``good'' initialization point $\vz^{(0)}$ (finding such a point is discussed in Section \ref{sec:spectral_init}). The solution is then updated iteratively via \emph{Wirtinger flow} (WF):
\begin{align}
\label{eq:grad_des_ud}
\vz^{(t+1)} = \vz^{(t)}-\eta\nabla f(\vz)\,,
\end{align}
where $\eta>0$ is some suitable step size, and $\nabla f(\vz)$ can be computed as
\begin{align}
\label{eq:gradient_1}
\begin{split}
\nabla f(\vz) =\left(\frac{\partial f}{\partial \vz}\right)^*=\frac{1}{m}\sum_{i=1}^m&\left(\vz^*\vA_i^*\vz-\vx^*\vA_i^*\vx\right)\vA_i\vz\\
&+\left(\vz^*\vA_i\vz-\vx^*\vA_i\vx\right)\vA_i^*\vz\,.
\end{split}
\end{align}
If $\vx$ is a global minimum of $f(\vz)$, then $\vx e^{\vj\phi}$ is also a global minimum for any $\phi\in(0,2\pi]$. Consequently, it is standard to define the squared distance between the recovered solution $\vz$ and the true solution $\vx$ as
\begin{align}
\label{eq:compute_min_dist}
\begin{split}
\mathrm{dist}^2\left(\vz, \vx\right)&=\min_{\phi\in(0,2\pi]}\left\|\vz-\vx e^{\vj\phi}\right\|_2^2\\
&=\|\vz\|_2^2+\|\vx\|_2^2-2|\vz^*\vx|\,,
\end{split}
\end{align}
where $\vz^* \vx = | \vz^* \vx | e^{\vj \phi_{\vz^* \vx}}$ and the minimum is achieved when $\phi = \phi_{\text{min}} := -\phi_{\vz^*\vx}$.

\subsection{Spectral initialization}
\label{sec:spectral_init}
Spectral initialization is widely used in problems with quadratic measurements to obtain an initialization that is close to a global optimum. Similar to \cite{Netrapalli2015:RPAM,Candes2015PhaseRV}, we show that the spectral initializer, $\vz^{(0)}$, is close to a global optimizer $\vx$ with high probability and can be used to initialize the WF update in \eqref{eq:grad_des_ud}. The rationale behind the spectral initialization strategy is that we can get a good estimate $\vS$ of $2\vx\vx^*$ using sufficient measurements. The spectral initializer $\vz^{(0)}$ can then be constructed from the leading left or right singular vectors of $\vS$.

\begin{figure*}[tpb]
    \centering
    \includegraphics[width=0.8\textwidth]{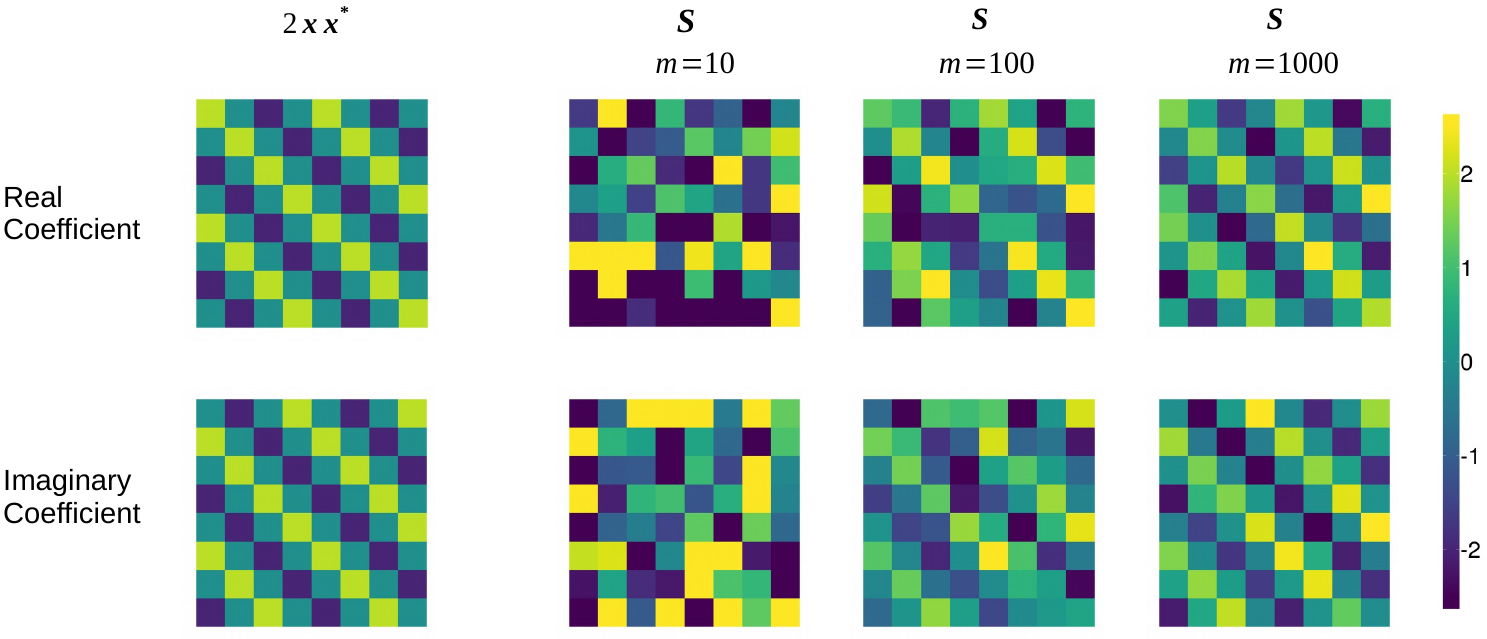}
    \caption{The complex signal to be recovered is $\vx=\left[1\ e^{\frac{\pi}{2}\vj}\ e^{\pi\vj}\ e^{\frac{3\pi}{2}\vj}\ 1\ e^{\frac{\pi}{2}\vj}\ e^{\pi\vj}\ e^{\frac{3\pi}{2}\vj}\right]^T$. We estimate $2\vx\vx^*$ with varying number of measurements. The estimate $\vS$ becomes increasingly accurate with more measurements.}
    \label{fig:concen_spec}
\end{figure*}

Unlike in the phase retrieval problem, which uses the Hermitian matrix $\frac{1}{m}\sum_{i=1}^my_i\va_i\va_i^*$ as the estimate of $\vI+2\vx\vx^*$ ($\vI$ is the identity matrix), we propose the following estimation of $2\vx\vx^*$ under the rotation-invariant sub-Gaussian measurement model
\begin{align}
\label{eq:spectral_S}
    \vS=\frac{1}{m}\sum_{i=1}^m\overline{y}_i\vA_i\,,
\end{align}
where $\overline{y}_i=\vx^*\vA_i^*\vx$ is the complex conjugate of $y_i$. To understand this intuitively, note that the expectation of $\vS$ is
\begin{align}
\label{eq:expectation_s_first}
\begin{split}
\mathbb{E}\left[\vS\right]=\sum_{rc}x_r\overline{x}_c\cdot\mathbb{E}\left[\frac{1}{m}\sum_{i=1}^m\overline{\vA}_{i,rc}\cdot\vA_i\right]\,.
\end{split}
\end{align}
Let $A_{i,rc}^{(R)}$ and $A_{i,rc}^{(I)}$ denote the real and imaginary coefficients of the $(r,c)$-th entry $\vA_{i,rc}$. Since all the coefficients of the matrix entries in $\vA_i$ are generated from a rotation-invariant distribution, any two coefficients are also rotation-invariant, pairwise uncorrelated and ``by assumption'' have unit variance. Using \cite[Proposition 4.1.1]{Bryc1995}, we have:
\begin{align}
\mathrm{Pr}\left(A_{i, rc}^{(R)}\in\mathbb{R}\right) &= \mathrm{Pr}\left(-A_{i,rc}^{(R)}\in\mathbb{R}\right)\\
\mathrm{Pr}\left(A_{i,rc}^{(R)}A_{i,kl}^{(R)}\in\mathbb{R}\right)&=\mathrm{Pr}\left(-A_{i,rc}^{(R)}A_{i,kl}^{(R)}\in\mathbb{R}\right),
\end{align}
where $(r,c)\neq(k,l)$. It is easy to verify that $\mathbb{E}\left[A_{i, rc}^{(R)}\right]=\mathbb{E}\left[A_{i,kl}^{(R)}\right]=0$ and $\mathbb{E}\left[A_{i,rc}^{(R)}A_{i,kl}^{(R)}\right]=0$ so that
\begin{align}
\label{eq:ri_uncorr}
\mathbb{E}\left[A_{i,rc}^{(R)}A_{i,kl}^{(R)}\right]=\mathbb{E}\left[A_{i, rc}^{(R)}\right]\cdot\mathbb{E}\left[A_{i,kl}^{(R)}\right]=0\,,
\end{align}
meaning that the coefficients $A_{i,rc}^{(R)}$ and $A_{i,kl}^{(R)}$ are pairwise uncorrelated, and so are ``any'' two different coefficients of the entries in $\vA_i$. Using \eqref{eq:ri_uncorr}, it is easy to check that
\begin{align}
\mathbb{E}\left[\overline{\vA}_{i,rc}\cdot\vA_{i,rc}\right]&=2\\
\mathbb{E}\left[\overline{\vA}_{i,rc}\cdot\vA_{i,kl}\right]&= 0\,,
\end{align}
where $r\neq k$ or $c\neq l$. Hence $\mathbb{E}\left[\frac{1}{m}\sum_{i=1}^m\overline{\vA}_{i,rc}\vA_i\right]$ is a matrix with the $(r,c)$-th entry equaling to $2$ and the rest of the entries being zeros. Thus
\begin{align}
\label{eq:exp_spec_init}
\mathbb{E}\left[\vS\right]=2\vx\vx^*\,.
\end{align}
The following lemma implies that the matrix $\vS$ concentrates around $\mathbb{E}\left[\vS\right]$ in spectral norm with high probability when $m$ is sufficiently large.

\begin{lemma}
\label{lemma:spectral_norm_concen}
Under the rotation-invariant sub-Gaussian measurement model given by \eqref{eq:quad_eqs}, for every $\nu > 0$, when the number of measurements satisfies $m > C n$ for some sufficiently large constant $C := C(\nu)$, we have for fixed unit vectors $\vp,\vq\in\mathbb{C}^n$ that
\begin{align}
\label{eq:spectral_norm_concen}
    \left\|\frac{1}{m}\sum_{i=1}^m\vp^*\vA_i^*\vq\cdot \vA_i - 2\vq\vp^* \right\|<\nu\,,
\end{align}
with probability at least $1-20\exp\big(-m\cdot C_1(C,\nu)\big)$, where $C_1(C,\nu)>0$ is some constant depending on $C$ and $\nu$.
\end{lemma}

Note that the statement of Lemma \ref{lemma:spectral_norm_concen} is slightly more general than what we need right now, since it allows $\vp \neq \vq$ (here we set both to $\frac{\vx}{\|\vx\|_2}$). This will be useful in the later sections. As derived in Appendix \ref{proof:lemma:spectral_norm_concen}, the concentration proof hinges on the rotation invariance of the measurement matrices $\vA_i$ \cite{Bryc1995}: If we define $\vB = \vR \vA$ with $\vR\in\mathbb{C}^{n\times n}$ being a complex unitary matrix, the real and imaginary coefficients of the entries in $\vB$ have the same joint distribution as those of $\vA$.

Lemma \ref{lemma:spectral_norm_concen} implies that $\vS$ is close to $2\vx\vx^*$ already for ``reasonable'' finite values of $m$ (see Fig. \ref{fig:concen_spec}), and the likelihood that this is not the case decays exponentially with $m$. Let $\{\vu_0,\bv_0\}$ be the ``leading'' left and right singular vectors of $\vS$, both $\vu_0$ and $\bv_0$ are highly correlated with $\frac{\vx}{\|\vx\|_2}$ with high probability, which is made precise in the proof of the below Lemma \ref{lemma:initialization}. We can then use either $\vu_0$ or $\bv_0$ to construct the spectral initializer $\vz^{(0)}$. Here we shall pick $\bv_0$ in the following discussion.
\begin{enumerate}
\item When the norm of $\vx$ is known and fixed, the spectral initializer is
\begin{align}
    \vz^{(0)}=\|\vx\|_2\cdot\bv_0\,.
\end{align}

\item When the norm of the signal is unknown, we can estimate it from the quadratic measurements. Using \eqref{eq:exp_spec_init}, we compute the following:
\begin{align}
\label{eq:unknown_norm}
\mathbb{E}\left[\frac{1}{2m}\sum_{i=1}^m\overline{y}_i y_i\right]=\mathbb{E}\left[\frac{1}{2}\vx^*\vS\vx\right]= \|\vx\|_2^4\ .
\end{align}
When $m$ is sufficiently large, we prove that $\frac{1}{2m} \sum_{i=1}^m \overline{y}_i y_i$ is close to its expectation $\|\vx\|_2^4$ with high probability (see the proof of Lemma \ref{lemma:initialization} in Appendix \ref{proof:lemma:initialization}). Based on this result, we can scale one of the leading singular vectors $\bv_0$ of $\vS$ to get our spectral initializer,
\begin{align}
\vz^{(0)} = \left(\frac{1}{2m} \sum_{i=1}^m \overline{y}_i y_i\right)^\frac{1}{4}\cdot\bv_0\,.
\end{align}
\end{enumerate}

Since $\bv_0$ is also the leading eigenvector of $\vS^*\vS$ we can use the power iteration to compute it and avoid a full singular value decomposition (SVD) of $\vS$. We initialize it with some random unit-norm vector, $\bv_0^{(0)}$, and compute the following power iteration until convergence
\begin{align}
    \label{eq:power_method_init}
    \bv_0^{(t+1)}=\frac{\vS^*\vS\bv_0^{(t)}}{\|\vS^*\vS\bv_0^{(t)}\|_2}\,,
\end{align}
whose computational complexity is $\mathcal{O}(n^2)$, as opposed to the $\mathcal{O}(n^3)$ complexity of a full SVD.

One interpretation of spectral initialization is that we are computing an approximate least-squares estimate of the matrix $\vx\vx^*$ in the subspace spanned by the measurement matrices. With this guiding principle we can adapt the spectral initialization to other types of measurements. For example, in our work on the uDGP \cite{Huang2018uDGP} we constructed an orthonormal basis for the matrix subspace spanning the measurement matrices. The measurements can be interpreted as projections of $\vx\vx^*$ on to this basis. The least-squares estimate of $\vx \vx^*$ and the corresponding spectral initialization can then be easily obtained.

The following lemma states that the distance between the spectral initializer $\vz^{(0)}$ and a global optimizer $\vx$ is small with high probability when $m$ is sufficiently large. 
\begin{lemma}
\label{lemma:initialization}
Under the rotation-invariant sub-Gaussian measurement model given by \eqref{eq:quad_eqs}, when the number of complex quadratic measurements satisfies $m>Cn$ for some sufficiently large constant $C$, for every $\delta\in(0,24)$, there exists a global optimizer $\vx$ of the loss in \eqref{eq:first_obj} such that the distance between the spectral initializer $\vz^{(0)}$ and $\vx$ obeys
\begin{align}
\label{eq:spec_init_ub_general}
\mathrm{dist}^2\left(\vz^{(0)}, \vx\right) \leq \frac{51}{24}\delta\|\vx\|_2^2\,,
\end{align}
with probability at least $1-20\exp\big(-m\cdot C_1(C,\delta)\big)$, where $C_1(C,\delta)>0$ is some constant depending on $C$ and $\delta$.
\end{lemma}

\begin{figure}[tpb]
    \centering
    \includegraphics[width=\columnwidth]{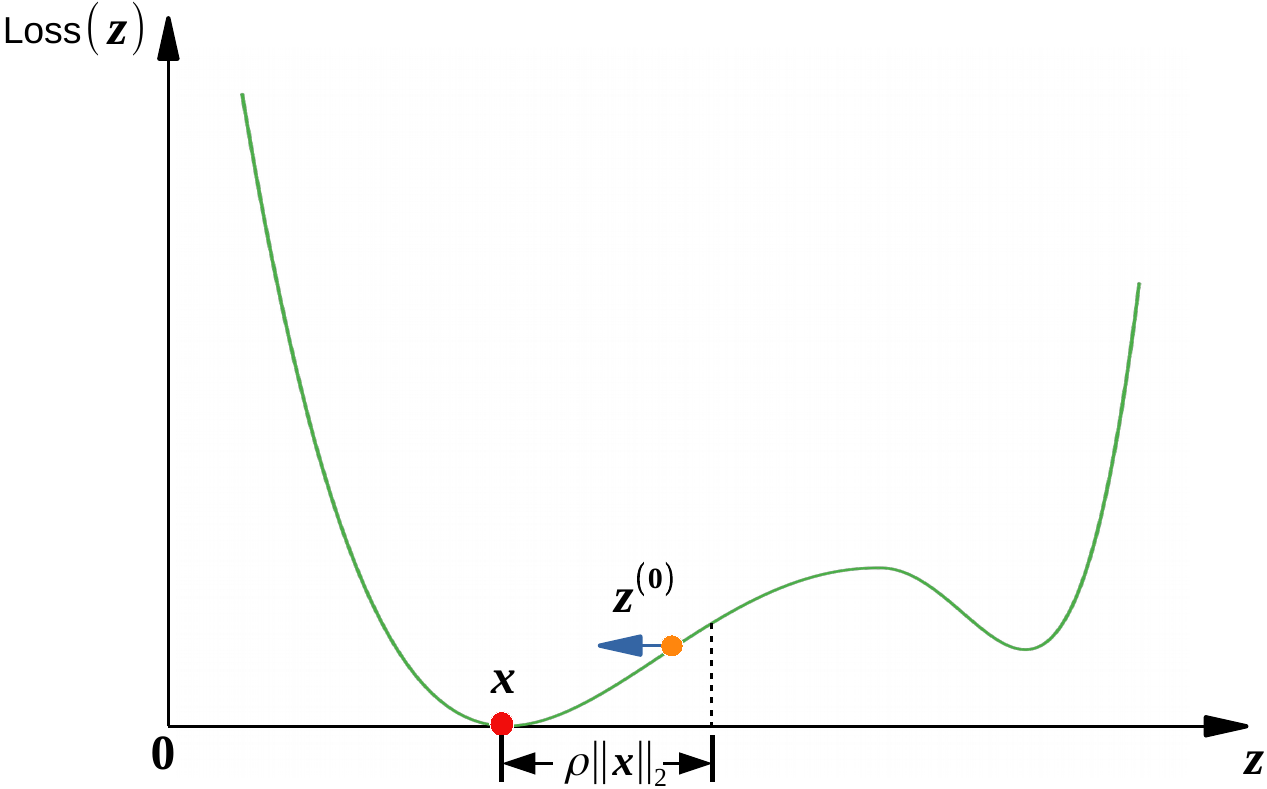}
    \caption{The loss function, $\mathrm{Loss}(\vz)$, gives rise to the basin of attraction around a global optimizer $\vx$, where Wirtinger flow updates can be used to recover $\vx$ with high probability when there are sufficient measurements.}
    \label{fig:convergence_neighborhood}
\end{figure}

As a consequence of Lemma \ref{lemma:initialization}, no matter how small a $\delta$ we choose, we can guarantee that $\vz^{(0)}$ is $\mathcal{O}(\delta)$-close to a global optimizer $\vx$ with high probability by increasing the number of measurements. As illustrated in Fig. \ref{fig:convergence_neighborhood}, suppose $\rho\|\vx\|_2$ is the size of the neighborhood around $\vx$ where a simple local optimization method such as WF can be used to recover $\vx$ with high probability. Such a neighborhood around $\vx$ is referred to as ``basin of attraction'' in \cite{Chen2015Trunc}. Our goal is then to balance the trade-off between making $\delta$ small enough so that $\vz^{(0)}$ falls within the basin of attraction and reducing the required number of measurements.

\section{Convergence analysis}
\label{sec:convergence_analysis}
Let $\vx_0$ denote a global optimizer, and $\mathcal{P}$ the set of all vectors that differ from $\vx_0$ by some phase shift $\phi$:
\[\mathcal{P}=\left\{\vx_0 e^{\vj\phi}:\ \phi\in(0,2\pi]\right\}.\]
In order to determine the neighborhood size $\rho\|\vx\|_2$, we study the convergence behavior of the WF iterates in the neighborhood $E(\rho)$ of $\mathcal{P}$, defined as
\begin{align*}
    E(\rho)=\left\{\vz\left. \  | \ \mathrm{dist}(\vz,\vx)\leq\rho\|\vx\|_2,\ \vx\in\mathcal{P}\right.\right\}\,,
\end{align*}
where $\mathrm{dist}(\vz,\vx)=\|\vz-\vx e^{\vj\phi_{\min}}\|_2$ is computed as in \eqref{eq:compute_min_dist}. The objective function $f(\vz)$ is said to satisfy the regularity condition $RC(\alpha,\beta,\rho)$ if the following holds for all $\vz\in E(\rho)$ \cite{Candes2015PhaseRV},
\begin{align}
\label{eq:regularity_condition}
\textnormal{Re}\left(\langle\nabla f(\vz), \vz-\vx e^{\vj\phi_{\min}}\rangle\right)\geq\frac{1}{\alpha}\textnormal{dist}^2(\vz,\vx)+\frac{1}{\beta}\|\nabla f(\vz)\|_2^2\,,
\end{align}
for the choice of constants $\alpha>0$, $\beta>0$, $\rho>0$. The regularity condition $RC(\alpha,\beta,\rho)$ can be derived straightforwardly by demanding that the WF step takes us closer to a global optimizer. More precisely, it ensures that the WF iterate \eqref{eq:grad_des_ud} with a step size $\eta\in\left(0,\frac{2}{\beta}\right]$ converges linearly to a global optimizer $\vx$ when the descent is initialized within the neighborhood $E(\rho)$ \cite[Lemma 7.10]{Candes2015PhaseRV}:
\begin{align}
\label{eq:wf_7_10}
\mathrm{dist}^2\left(\vz^{(t)},\vx\right)\leq\left(1-\frac{2\eta}{\alpha}\right)^t\mathrm{dist}^2\left(\vz^{(0)},\vx\right)\,.
\end{align}

One of the main challenges in going from the real to the complex case lies in the more complicated definition of the regularity condition and the related convergence analysis. Note that in the real case the regularity condition reads simply
\[
    \langle\nabla f(\vz), \vz-\vx\rangle\geq\frac{1}{\alpha}\textnormal{dist}^2(\vz,\vx)+\frac{1}{\beta}\|\nabla f(\vz)\|_2^2\,.
\]
Following the strategy from the Wirtinger flow paper \cite{Candes2015PhaseRV} to lower-bound the left-hand side of \eqref{eq:regularity_condition}
would result in complicated derivations involving the computation of the Hessian matrix. We show below how, thanks to our measurement model, we can greatly simplify these derivations using the central Lemma \ref{lemma:spectral_norm_concen_all}.

\subsection{Establishing the convergence criterion}

We now show that there exist choices of parameters $\alpha, \beta, \rho$ such that the objective function $f(\vz)$ introduced in \eqref{eq:first_obj} satisfies the regularity condition $RC(\alpha,\beta, \rho)$ in \eqref{eq:regularity_condition} with high probability, and choose a set of parameter values such that spectral initialization followed by WF succeeds with high probability. The existence of good parameters is shown in three steps according to \eqref{eq:regularity_condition} and \eqref{eq:wf_7_10}:
\begin{enumerate}
\item Finding a positive lower bound on $\mathrm{Re}\left(\langle\nabla f(\vz), \vz-\vx e^{\vj\phi_{\min}}\rangle\right)$;
\item Finding an upper bound on $\|\nabla f(\vz)\|_2^2$;
\item Choosing a suitable set of $(\alpha,\beta,\rho)$-values to obtain the regularity condition $RC(\alpha,\beta,\rho)$ in \eqref{eq:regularity_condition}.
\end{enumerate}

The main tool in proving these steps is a matrix concentration bound (a high-probability spectral norm bound) similar to the one in Lemma \ref{lemma:spectral_norm_concen}. However, Lemma \ref{lemma:spectral_norm_concen} is stated for a particular, fixed choice of the unit vectors $\vp,\vq$. Since we want the above bounds which imply the regularity condition to hold \underline{for all} vectors in $E(\rho)$, it will be useful to strengthen Lemma \ref{lemma:spectral_norm_concen} so that it holds simultaneously for all choices of $\vp$ and $\vq$.

\begin{lemma}
\label{lemma:spectral_norm_concen_all}
Under the rotation-invariant sub-Gaussian measurement model given by \eqref{eq:quad_eqs}, for every $\nu > 0$, when the number of measurements $m$ satisfies $m > C n$ for some sufficiently large constant $C := C(\nu)$, we have for all $\vp,\vq\in\mathbb{C}^n$ satisfying $\|\vp\|_2=1$, $\|\vq\|_2=1$ that
\begin{align}
\label{eq:spectral_norm_concen_all}
    \left\|\frac{1}{m}\sum_{i=1}^m\vp^*\vA_i^*\vq\cdot \vA_i - 2\vq\vp^* \right\|<\nu\,,
\end{align}
with probability at least $1-20\exp\big(-m\cdot C_2(C,\nu)\big)$, where $C_2(C,\nu)>0$ is some constant depending on $C$ and $\nu$.
\end{lemma}

In the standard phase retrieval measurement model \cite{Candes2015PhaseRV}, rows and columns of the measurement matrices are correlated. This is in particular the case for the $(r,c)$-th entry $\overline{a_{i,r}}a_{i,c}$ and the $(c,r)$-th entry $\overline{a_{i,c}}a_{i,r}$. We have
\begin{align}
    \mathbb{E}\left[\overline{a_{i,r}}a_{i,c}\cdot\overline{a_{i,c}}a_{i,r}\right]=\mathbb{E}\left[|a_{i,r}|^2|a_{i,c}|^2\right]>0\,. 
\end{align}
As a consequence, the distribution of $\overline{a_{i,r}}a_{i,c}\cdot\overline{a_{i,c}}a_{i,r}$ is not centered. This precludes a result parallel to Lemma \ref{lemma:spectral_norm_concen_all} which in our case lets us establish the regularity condition in a straightforward way.

With Lemma \ref{lemma:initialization} and Lemma \ref{lemma:spectral_norm_concen_all} in hand, we can now state our main result.
\begin{theorem}
\label{theorem:main_theorem}
Under the rotation-invariant sub-Gaussian measurement model given by \eqref{eq:quad_eqs}, when the number of complex quadratic measurements $m > Cn$ for some sufficiently large constant $C$,

\begin{enumerate}
    \item There exists a choice of  $1>\nu > 0,\ 1>\rho>0,\ \alpha > 0$, and $\beta > 0$, such that $RC(\alpha, \beta, \rho)$ holds with probability at least $1-\kappa^\prime\cdot\exp\big(-m\cdot C_2(C,\nu)\big)$, where $C_2(C,\nu)$ is some constant depending $C$ and $\nu$, and $\kappa^\prime>0$ is an absolute constant.
    \item Under this choice of parameters, if the step size $\eta$ is chosen so that $0 <\eta \leq \frac{2}{\beta}$, the WF iterates \eqref{eq:grad_des_ud} initialized at the spectral initializer $\vz^{(0)}$ converge linearly to a global optimizer $\vx$,
    \begin{align}
    \label{eq:thm_general}
    \mathrm{dist}^2\left(\vz^{(t)}, \vx\right) \leq \left(1-\frac{2\eta}{\alpha} \right)^t\cdot\rho^2\|\vx\|_2^2 \,,
    \end{align}
    with probability at least $1-\kappa\cdot\exp\big(-m\cdot C_3(C,\nu,\rho)\big)$, where $\kappa>0$ is some absolute constant and $C_3(C,\nu,\rho)$ depends on $C,\nu$ and $\rho$, but not on $m$..
\end{enumerate}
\end{theorem}

\begin{proof}
The main task in proving the theorem is to prove Part 1 (the regularity condition). Once we establish that there exists a choice of parameters such that the regularity condition holds in the neighborhood of a global minimizer, Part 2 (linear convergence to a global minimizer) follows simply by noting that $RC(\alpha, \beta, \rho)$ implies
\[
    \mathrm{dist}^2(\vz^{(t+1)}, \vx) \leq \left( 1 - \frac{2\eta}{\alpha} \right) \mathrm{dist}^2(\vz^{(t)}, \vx),
\]
whenever $\vz^{(t)} \in E(\rho)$.

\noindent{\bfseries Part 1}: Establishing the regularity condition. We work as follows:

\paragraph{Finding a positive lower bound on $\mathrm{Re}\left(\langle\nabla f(\vz), \vz-\vx e^{\vj\phi_{\min}}\rangle\right)$ in the neighborhood of $\vx$} We let $\vh=\vz e^{-\vj\phi_{\min}}-\vx$ so that
\begin{align}
\label{eq:pfthm1-Rrhnablaf}
\mathrm{Re}\left(\langle\vh,\nabla f(\vz)\rangle\right) = \mathrm{Re}(\langle\nabla f(\vz), \vz-\vx e^{\vj\phi_{\min}}\rangle),
\end{align}
and equivalently look for a lower bound on $\mathrm{Re}\left(\langle\vh,\nabla f(\vz)\rangle\right)$ for all $\vh,\vx\in\mathbb{C}^n$ satisfying $\frac{\|\vh\|_2}{\|\vx\|_2}\leq\rho$. We proceed by showing that $\mathbb{E}\left[\mathrm{Re}\left(\langle\vh,\nabla f(\vz)\rangle\right)\right]>0$, and that $\mathrm{Re}\left(\langle\vh,\nabla f(\vz)\rangle\right)$ is close to $\mathbb{E}\left[\mathrm{Re}\left(\langle\vh,\nabla f(\vz)\rangle\right)\right]$ with high probability, so that a strictly positive lower bound can be established with high probability when $m$ is sufficiently large. The expression \eqref{eq:pfthm1-Rrhnablaf} can be expanded as
\begin{align}
\label{eq:first_obj_inner_prod}
\begin{split}
&\textnormal{Re}\left(\langle\vh,\nabla f(\vz)\rangle\right)\\
&=\frac{1}{m}\sum_{i=1}^m\left[2 \cdot \textnormal{Re}\left(\vh^*\vA_i^*\vh\cdot\vh^*\vA_i\vh\right)+3\cdot\textnormal{Re}\left(\vh^*\vA_i^*\vh\cdot\vh^*\vA_i\vx\right) \right. \\
&\quad\quad\quad+3\cdot\textnormal{Re}\left(\vh^*\vA_i^*\vx\cdot\vh^*\vA_i\vh\right)+2\cdot\textnormal{Re}\left(\vh^*\vA_i^*\vx\cdot\vh^*\vA_i\vx\right) \\
&\quad\quad\quad\left. +\ \textnormal{Re}\left(\vh^*\vA_i^*\vx \cdot \vx^*\vA_i\vh\right) + \ \textnormal{Re}\left(\vx^*\vA_i^*\vh\cdot\vh^*\vA_i\vx\right)\right]\,.
\end{split}
\end{align}
We rely on the spectral norm bound in Lemma \ref{lemma:spectral_norm_concen_all} to lower-bound \eqref{eq:first_obj_inner_prod}, by bounding each of the six terms in turn. Since the exact same logic applies to all terms, we give details only for the second one, $\frac{1}{m}\sum_{i=1}^m\textnormal{Re}(\vh^*\vA_i^*\vx\cdot\vh^*\vA_i\vx)$. Let $\vh=\|\vh\|_2\cdot\widehat{\vh}$, and $\vx=\|\vx\|_2\cdot\widehat{\vx}$, where $\|\widehat{\vh}\|_2=1$ and $\|\widehat{\vx}\|_2=1$. By Lemma \ref{lemma:spectral_norm_concen_all} we have 
\begin{align}
\label{eq:expand}
\begin{split}
    &\left|\frac{1}{m}\sum_{i=1}^m\textnormal{Re}(\vh^*\vA_i^*\vx\cdot\vh^*\vA_i\vx)-2\cdot\textnormal{Re}(\vh^*\vx\vh^*\vx)\right|\\
    &\leq \left|\frac{1}{m}\sum_{i=1}^m\vh^*\vA_i^*\vx\cdot\vh^*\vA_i\vx-2\vh^*\vx\vh^*\vx\right|\\
    &\leq\|\vh\|_2^2\|\vx\|_2^2\cdot\left\|\frac{1}{m}\sum_{i=1}^m\widehat{\vh}^*\vA_i^*\widehat{\vx}\cdot\vA_i-2\widehat{\vx}\widehat{\vh}^*\right\|\\
    &\leq\nu\|\vh\|_2^2\|\vx\|_2^2\,,
\end{split}
\end{align}
\underline{for all} $\vh$ and $\vx$ with probability at least $1-20\exp\big(-m\cdot C_2(C,\nu)\big)$. From \eqref{eq:compute_min_dist} we see that $\vh^*\vx=\vz^*\vx e^{\vj\phi_{\min}}-\|\vx\|_2^2=|\vz^*\vx|-\|\vx\|_2^2$ is a real number. We then have the following
\begin{align}
\label{eq:expand_left}
    \frac{1}{m}\sum_{i=1}^m\textnormal{Re}(\vh^*\vA_i^*\vx\cdot\vh^*\vA_i\vx) \geq 2(\vh^*\vx)^2-\nu\|\vh\|_2^2\|\vx\|_2^2\,,
\end{align}
also holds \underline{for all} $\vh$ and $\vx$ with probability\footnote{The probability that \eqref{eq:expand_left} holds is no less than the probability that \eqref{eq:expand} holds.} at least $1-20\exp\big(-m\cdot C_2(C,\nu)\big)$. Repeating for every term in \eqref{eq:first_obj_inner_prod}, for all $\vh,\vx\in\mathbb{C}^n$ satisfying $\frac{\|\vh\|_2}{\|\vx\|_2}\leq\rho$, we find that
\begin{equation}
\label{eq:lower_bound_re}
\begin{aligned}
&\textnormal{Re}\left(\langle\vh,\nabla f(\vz)\rangle\right)\\
&\geq 4\left[\left(\|\vh\|_2^2+\vh^*\vx\right)^2+\|\vh\|_2^2\left(\|\vx\|_2^2+\vx^*\vh\right)\right]\\
&\quad-2\|\vh\|_2^2\nu\left(\|\vh\|_2^2+3\|\vh\|_2\|\vx\|_2+2\|\vx\|_2^2\right)\\
&\geq 4\left[\|\vh\|_2^2\left(\|\vx\|_2^2+\vx^*\vh\right)\right]\\
&\quad-2\|\vh\|_2^2\nu\left(\|\vh\|_2^2+3\|\vh\|_2\|\vx\|_2+2\|\vx\|_2^2\right)\\
&\geq 4\left(1-\rho-\frac{\nu}{2}\left(2+3\rho+\rho^2\right)\right)\cdot\|\vh\|_2^2\|\vx\|_2^2\\
&=c_1(\nu,\rho)\cdot\|\vh\|_2^2\|\vx\|_2^2\,,
\end{aligned}
\end{equation}
holds with probability at least $1-\kappa_1\exp\big(-m\cdot C_2(C,\nu)\big)$ where $\kappa_1>0$ is some absolute constant.

\paragraph{Finding an upper bound on $\|\nabla f(\vz)\|_2^2$ in the neighborhood of $\vx$} We can rewrite $\|\nabla f(\vz)\|_2$ as follows,
\begin{align}
\label{eq:ub_derivative_norm}
\begin{split}
    &\|\nabla f(\vz)\|_2\\
    &=\left\|\frac{1}{m}\sum_{i=1}^m\left(\vh^*\vA_i^*\vh+\vh^*\vA_i^*\vx+\vx^*\vA_i^*\vh\right)\vA_i(\vh+\vx)\right.\\
    &\quad\left.\phantom{\frac{1}{m}} + \left(\vh^*\vA_i\vh+\vh^*\vA_i\vx+\vx^*\vA_i\vh\right)\vA_i^*(\vh+\vx)\right\|_2.
\end{split}
\end{align}
To upper-bound \eqref{eq:ub_derivative_norm} we again rely on the spectral norm bound in Lemma \ref{lemma:spectral_norm_concen_all}. Let $\vh=\|\vh\|_2\cdot\widehat{\vh}$, and $\vx=\|\vx\|_2\cdot\widehat{\vx}$, where $\|\widehat{\vh}\|_2=1$ and $\|\widehat{\vx}\|_2=1$. We bound the second term (say) in \eqref{eq:ub_derivative_norm} as follows
\begin{align}
\begin{split}
&\left\|\frac{1}{m}\sum_{i=1}^m\vh^*\vA_i^*\vx\cdot\vA_i(\vh+\vx)\right\|_2-\left\|2\vx\vh^*(\vh+\vx)\right\|_2\\
&\leq \left\|\left(\frac{1}{m}\sum_{i=1}^m\vh^*\vA_i^*\vx\cdot\vA_i-2\vx\vh^*\right)(\vh+\vx)\right\|_2\\
&\leq \left\|\frac{1}{m}\sum_{i=1}^m\widehat{\vh}^*\vA_i^*\widehat{\vx}\cdot\vA_i-2\widehat{\vx}\widehat{\vh}^*\right\|\cdot\|\vh\|_2\|\vx\|_2\cdot\|\vh+\vx\|_2\\
&\leq\nu\|\vh\|_2\|\vx\|_2\cdot\left(\|\vh\|_2+\|\vx\|_2\right)\,,
\end{split}
\end{align}
holds with probability at least $1-20\exp\big(-m\cdot C_2(C,\nu)\big)$, implying that
\begin{align}
\begin{split}
    &\left\|\frac{1}{m}\sum_{i=1}^m\vh^*\vA_i^*\vx\cdot\vA_i(\vh+\vx)\right\|_2\\
    &\leq (\nu+2)\|\vh\|_2\|\vx\|_2\cdot\left(\|\vh\|_2+\|\vx\|_2\right)\,,
\end{split}
\end{align}
holds with at least the same probability. Repeating for all terms in \eqref{eq:ub_derivative_norm}, we get
\begin{align}
\label{eq:upper_bound_norm}
\begin{split}
    \|\nabla f(\vz)\|_2^2 &\leq 4(2+\nu)^2\left(\rho^2+3\rho+2\right)^2\cdot\|\vh\|_2^2\|\vx\|_2^4\\
    &=c_2(\nu,\rho)\cdot\|\vh\|_2^2\|\vx\|_2^4\,,
\end{split}
\end{align}
with probability at least $1-\kappa_2\exp\big(-m\cdot C_2(C,\nu)\big)$ for all $\vh,\vx\in\mathbb{C}^n$ satisfying $\frac{\|\vh\|_2}{\|\vx\|_2}\leq\rho$, where $\kappa_2>0$ is some absolute constant.

\paragraph{Choosing suitable $\nu,\rho$ and $\alpha,\beta$}

We have that \eqref{eq:lower_bound_re} and \eqref{eq:upper_bound_norm} hold simultaneously with probability at least $1-\kappa'\cdot\exp\big(-m\cdot C_2(C,\nu)\big)$ where $\kappa'>0$ is some absolute constant. 

In \eqref{eq:lower_bound_re}, $c_1(\nu,\rho)>0$ is a sufficient condition to make $\mathrm{Re}(\langle\vh,\nabla f(\vz)\rangle)>0$ so that we could establish the regularity condition.
\begin{enumerate}
    \item If $\nu\geq1$, we have $c_1(\nu,\rho)\leq-\frac{5}{2}\rho-\frac{1}{2}\rho^2<0$. Hence $\nu$ needs to be less than $1$. 
    \item Given some $\nu\in(0,1)$, there exists a matching $\rho\in(0,1)$ to ensure $c_1(\nu,\rho)>0$. The chosen $\nu,\rho$ always lead to $c_2(\nu,\rho)>0$ in \eqref{eq:upper_bound_norm}.
\end{enumerate}
In other words, the radius of the convergence neighbourhood $E(\rho)$ could grow as large as $\|\vx\|_2$ when there are sufficient measurements.

It remains to show that there exist $\alpha>0$, $\beta>0$ so that 
\begin{align}
    \label{eq:regularity_condition_diff}
    \mathrm{Re}(\langle\vh,\nabla f(\vz)\rangle)\geq\frac{1}{\alpha}\|\vh\|_2+\frac{1}{\beta}\|\nabla f(\vz)\|_2^2\,,
\end{align}
holds with probability at least $1-\kappa'\cdot\exp\big(-m\cdot C_2(C,\nu)\big)$ for all $\vh,\vx\in\mathbb{C}^n$ satisfying $\frac{\|\vh\|_2}{\|\vx\|_2}\leq\rho$. Note that \eqref{eq:regularity_condition_diff} is equivalent to the regularity condition \eqref{eq:regularity_condition}.

The parameters $\alpha$ and $\beta$ can be chosen as follows.
\begin{enumerate}
\item From \eqref{eq:wf_7_10} we need $1-\frac{2\eta}{\alpha}\geq 0$. Since $\eta\in(0,\frac{2}{\beta}]$, according to \cite[Lemma 7.10]{Candes2015PhaseRV}, $\alpha$ and $\beta$ should satisfy $\frac{4}{\alpha\beta}\leq 1$. Making the change of variable $ \alpha = a / {\|\vx\|_2^2}$, $\beta = b \cdot c_2(\nu,\rho)\|\vx\|_2^2$, we then require
\begin{equation}
\label{eq:thm_cst_1}
\frac{4}{ab}\leq c_2(\nu,\rho)\,.
\end{equation}
\item From \eqref{eq:lower_bound_re} and \eqref{eq:upper_bound_norm}, the regularity condition \eqref{eq:regularity_condition} will hold if
\begin{align}
\begin{split}
c_1(\nu,\rho)\cdot\|\vh\|_2^2\|\vx\|_2^2
&\geq\frac{1}{\alpha}\|\vh\|_2^2+\frac{1}{\beta}\cdot c_2(\nu,\rho)\|\vh\|_2^2\|\vx\|_2^4\\
&=\left(\frac{1}{a}+\frac{1}{b}\right)\|\vh\|_2^2\|\vx\|_2^2\,,
\end{split}
\end{align}
or equivalently
\begin{align}
\label{eq:thm_cst_2}
\frac{1}{a}+\frac{1}{b}\leq c_1(\nu,\rho)\,.
\end{align}
\end{enumerate}
There exist many choices of $a$ and $b$ (and thus $\alpha$ and $\beta$) that simultaneously satisfy \eqref{eq:thm_cst_1} and \eqref{eq:thm_cst_2}, and consequently the regularity condition \eqref{eq:regularity_condition}. For example, to get the best convergence rate, we can choose $a$, $b$ that maximize $\frac{4}{\alpha\beta}=\frac{4}{c_2(\nu,\rho)}\cdot\frac{1}{ab}$ subject to \eqref{eq:thm_cst_1} and \eqref{eq:thm_cst_2}:
\begin{itemize}
\item If $c_2(\nu,\rho)\geq c_1(\nu,\rho)^2$, we can choose
\begin{align}
a =\frac{2}{c_1(\nu,\rho)}, \quad \quad
b =\frac{2}{c_1(\nu,\rho)}\,.
\end{align}
\item If $c_2(\nu,\rho)<c_1(\nu,\rho)^2$, we can choose $a$, $b$ that solve
\begin{align}
ab =\frac{4}{c_2(\nu,\rho)}, \quad \quad
a+b \leq\frac{4c_1(\nu,\rho)}{c_2(\nu,\rho)}\,.
\end{align}
\end{itemize}

\noindent{\bfseries Part 2}: Linear convergence to a global minimizer with spectral initialization.

After $\nu,\rho,\alpha$, and $\beta$ are chosen in Part 1, it remains to ensure that the spectral initializer $\vz^{(0)}$ falls inside the neighbourhood $E(\rho)$ with high probability. Using Lemma \ref{lemma:initialization}, we have that 
\begin{align}
    \label{eq:spec_init_reg_con}
    \frac{\mathrm{dist}^2\left(\vz^{(0)},\vx\right)}{\|\vx\|_2^2}\leq\rho^2\,,
\end{align}
holds with probability at least $1-20\exp\left(-m\cdot C_1\left(C,\frac{24}{51}\rho^2\right)\right)$. Note that choosing $\rho\in(0,1)$ naturally satisfies the constraint imposed on $\frac{24}{51}\rho^2$ in Lemma \ref{lemma:initialization}, i.e. $\frac{24}{51}\rho^2\in\left(0,24\right)$. Combining \eqref{eq:regularity_condition_diff} and \eqref{eq:spec_init_reg_con}, we have that the WF update \eqref{eq:grad_des_ud} linearly converges to a global minimizer with probability at least $1-\kappa\cdot\exp\big(-m\cdot C_3(C,\nu,\rho)\big)$, where $\kappa>0$ is some absolute constant and $C_3(C,\nu,\rho)=\min\left\{C_2(C,\nu),\  C_1(C,\frac{24}{51}\rho^2)\right\}$ depends on $\{C,\nu,\rho\}$, but not on $m$.

\end{proof}

\subsection{Choosing the \texorpdfstring{$\{\alpha,\beta,\rho\}$} --values for the regularity condition} 
\label{sec:choosing_reg_params}

Generally, for some $\nu\in(0,1)$, one begins by choosing a suitable $\rho\in(0,1)$ to obtain a positive lower bound on $\mathrm{Re}\left(\langle\nabla f(\vz), \vz-\vx e^{\vj\phi_{\min}}\rangle\right)$ as in \eqref{eq:lower_bound_re}. The values of $\alpha>0,\ \beta>0$ should be large enough to ensure the regularity condition holds for all $\vz\in E(\rho)$, and satisfy $\frac{4}{\alpha\beta}\leq 1$ so that the iterates $\vz^{(t)}$ come closer to $\vx$ with each iteration.

As derived in the proof of Theorem \ref{theorem:main_theorem}, it is clear that there are many choices for the values of $\{\alpha,\beta,\rho\}$ such that the initializer $\vz^{(0)}\in E(\rho)$ and the objective function $f(\vz)$ satisfies the regularity condition \eqref{eq:regularity_condition} with high probability. For concreteness, we shall showcase a particular choice of ``good'' parameter values.

If we choose $\nu=0.01$ and $\rho= 0.2$, it is easy to verify that
\begin{align}
\textnormal{Re}\left(\left<\vh,\nabla f(\vz)\right>\right)&>3|\vh\|_2^2\|\vx\|_2^2\\
\|\nabla f(\vz)\|_2^2&<120\|\vh\|_2^2\|\vx\|_2^4\,,
\end{align}
where $\vh=\vz e^{-\vj \phi_{\min}}-\vx$. We can then choose $\alpha=\frac{2}{3}\cdot\frac{1}{\|\vx\|_2^2}$ and $\beta=\frac{2}{3}\cdot 120\|\vx\|_2^2$ to obtain
\begin{align}
\begin{split}
\textnormal{Re}\left(\left<\vh,\nabla f(\vz)\right>\right)&> \frac{3}{2}\|\vh\|_2^2\|\vx\|_2^2 + \frac{3}{2}\|\vh\|_2^2\|\vx\|_2^2\\
&>\frac{1}{\alpha}\|\vh\|_2^2+\frac{1}{\beta}\|\nabla f(\vz)\|_2^2\,.
\end{split}
\end{align}
When $m$ is sufficiently large, the regularity condition holds for all $\vh$ satisfying $\rho=\frac{\|\vh\|_2}{\|\vx\|_2}\leq 0.2$ with high probability. We can show that the following update converges linearly to a global optimizer:
\begin{align}
\vz^{(t+1)} = \vz^{(t)}-\eta\nabla f(\vz)\,, \tag{\ref{eq:grad_des_ud} revisited}
\end{align}
where $0<\eta\leq\frac{2}{\beta}$.

\section{Experimental results}
\label{sec:experiments}
We perform numerical experiments to corroborate the theoretical results. Theorem \ref{theorem:main_theorem} states that the step size is upper-bounded by $\frac{2}{\beta}$ where $\beta$ is one of the regularity condition parameters in \eqref{eq:regularity_condition}. In Section \ref{sec:choosing_reg_params} $\beta$ is proportional to the squared norm of the signal to recover. Hence, in all experiments the step size is chosen as $\frac{0.1}{\|\vx\|_2^2}$ where the signal norm, $\|\vx\|_2^2$, is estimated using \eqref{eq:unknown_norm}. The value of $0.1$ was experimentally found to give a suitable balance between convergence speed and reliability.\footnote{Code available at \url{https://github.com/swing-research/random_quadratic_equations} under the MIT License.}

\subsection{Rotation-invariant sub-Gaussian distributions} \label{sec:rot_exp_family}
We next illustrate examples of rotation-invariant sub-Gaussian distributions which we use in our experiments. Let
\begin{align}
    \vu = \|\vs\|_2^q\cdot\vs\,,
\end{align}
where $q\in[-1,0]$, and $\vs \sim\mathcal{N}(\vzero, \vI_{d \times d})$, $\vs\neq\vzero$. If $q=-1$, the variable $\vu$ is uniformly distributed on the sphere $\|\vu\|_2=1$.
If $q\in(-1,0)$, the pdf of $\vu$ is
\begin{align}
\label{eq:rot_exp_pdf}
    p(\vu)=\frac{1}{q+1}\|\vu\|_2^{-\frac{qd}{q+1}}(2\pi)^{-\frac{d}{2}}\exp\left(-\frac{1}{2}\|\vu\|_2^{\frac{2}{q+1}}\right)\,,
\end{align}
(see Appendix \ref{app:ri_exp_fam}). If $q=0$, then $\vu=\vs$ follows the standard multivariate Gaussian distribution. In the experiments, we further scale $\vu\in\mathbb{R}^{d}$ with the scaling parameter $\gamma>0$: $\vr=\gamma\cdot\vu$ so that the coefficient vector $\vr$ satisfies $\mu^2=\text{Var}(r_i)=\mathbb{E}\left[r_i^2\right]=1$.

\begin{figure}[tpb]
    \centering
    \includegraphics[width=0.9\columnwidth]{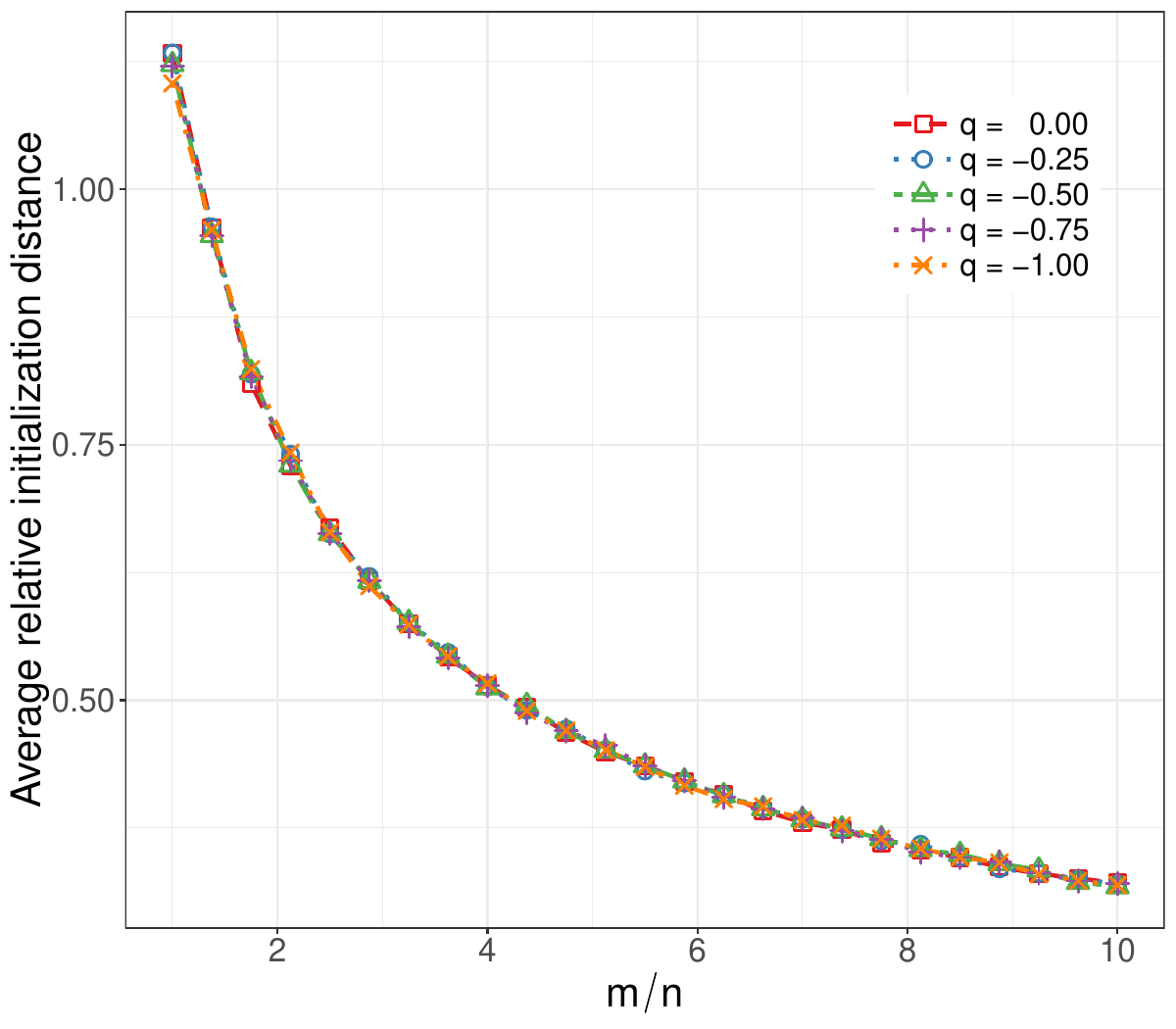}
    \caption{Closeness of spectral initialization with varying number of measurements where the complex random measurement matrices are from the rotation-invariant sub-Gaussian measurement model in Section \ref{sec:rot_exp_family}.}
    \label{fig:init_closeness_any_norm}
\end{figure}

\subsection{Closeness of spectral initializer}
\label{sec:exp_spec_init}
In this experiment we monitor how the distance between the initialization and the true solution varies with the number of measurements. We fix $n = 100$ and try different values of $m$ with $\frac{m}{n}$ uniformly sampled between $1$ and $10$. We run $100$ random trials for each $\frac{m}{n}$ value and calculate the average relative distance between the initialization and a global optimizer. In each trial we generate a random signal $\vx \in \mathbb{C}^n$ and $m$ complex random rotation-invariant sub-Gaussian matrices from the same distribution to produce $m$ complex quadratic measurements. We repeat this experiment with multiple distributions by varying the $q$ parameter in Section \ref{sec:rot_exp_family}.

Distance between complex signals is defined in \eqref{eq:compute_min_dist}. We define relative distance as $\frac{\mathrm{dist}\left(\vx, \vz^{(0)}\right)}{\|\vx\|_2}$ where $\vx$ is the original signal and $\vz^{(0)}$ is the initialization. In Fig. \ref{fig:init_closeness_any_norm} we can see that the spectral initializer comes closer to a global optimizer as $\frac{m}{n}$ increases. The behavior is not affected by varying $q$.

\subsection{Phase transition behavior}
\label{subsec:exp_phase_transition}
In this experiment we evaluate how the proposed approach transits from a failure phase to a success phase as we increase the number of measurements. We fix $n = 100$ and try different values of $m$ with $\frac{m}{n}$ sampled uniformly between $1.5$ and $5.5$. We again run 100 random trials for each $\frac{m}{n}$ value and calculate the success rate. Success is declared if the relative distance between the recovered and true signal is less than $10^{-5}$. Again, in each trial a random signal $\vx \in \mathbb{C}^n$ is reconstructed and multiple distributions from Section \ref{sec:rot_exp_family} are used to generate the measurement matrices.

The iterative WF reconstruction is terminated if the relative distance between successive iterations is less than $10^{-6}$ or if 2500 iterations are completed. Here we define relative distance between successive iterates as $\frac{\mathrm{dist}\left(\vz^{(t-1)},\ \vz^{(t)}\right)}{\|\vx\|_2}$ where $\vz^{(t-1)}$ and $\vz^{(t)}$ are two solutions recovered from successive WF iterates and $\vx$ is the original signal. As the true signal and its norm are unknown during the WF updates, \eqref{eq:unknown_norm} is used to estimate $\|\vx\|_2$. Fig. \ref{fig:phase_transition_any_norm} shows that approximately $4n$ measurements are needed to successfully recover the signal for all tested values of $q$.\footnote{We note that changing the numerical tolerance in the algorithm stopping criterion can shift the curve.}

\begin{figure}[tpb]
    \centering
    \includegraphics[width=0.9\columnwidth]{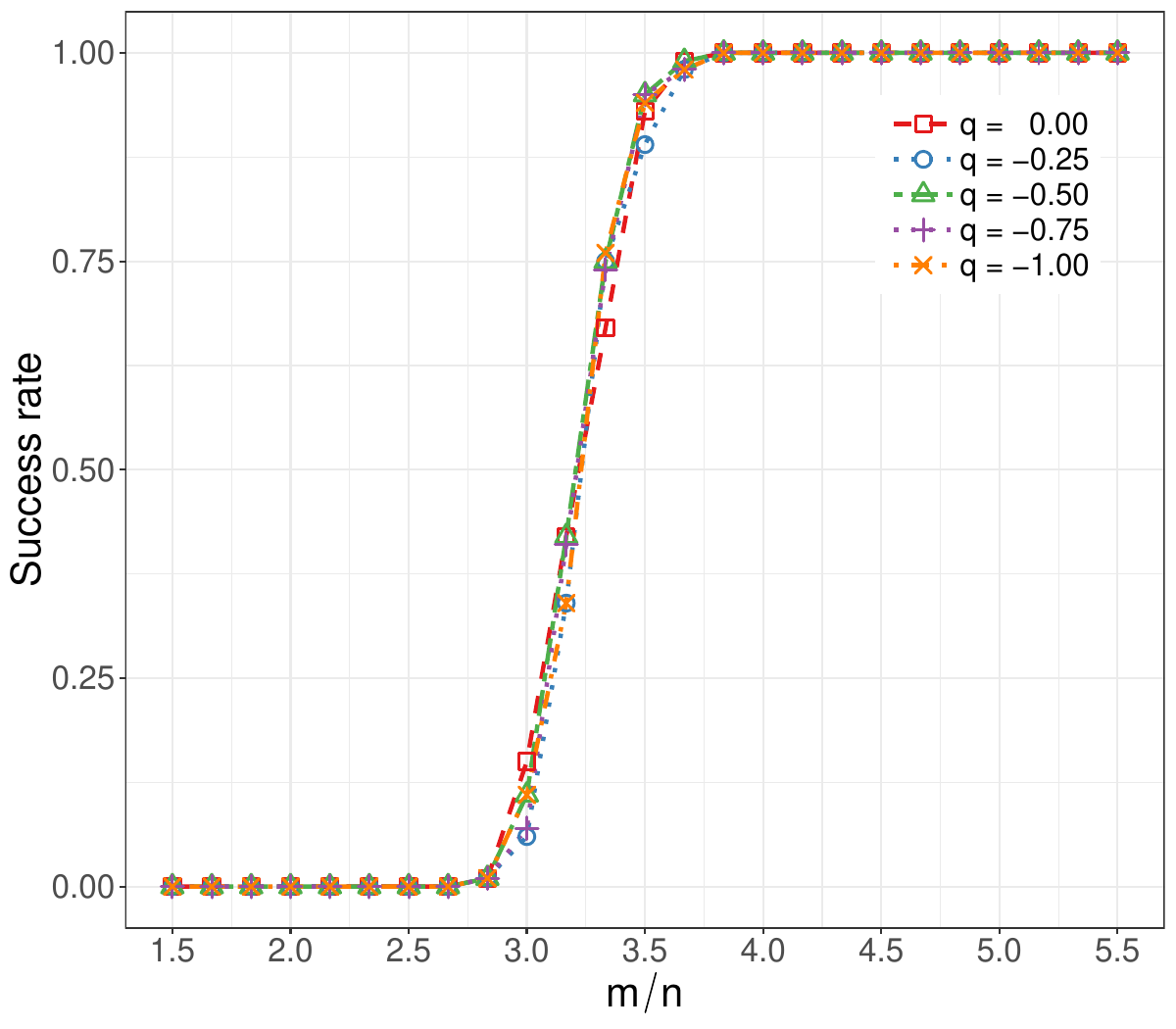}
    \caption{Success transition plot showing the empirical probability of success based on 100 trials with varying number of measurements where the complex random measurement matrices are from the rotation-invariant sub-Gaussian model in Section \ref{sec:rot_exp_family}.}
    \label{fig:phase_transition_any_norm}
\end{figure}

\subsection{Reconstruction of an image}
In this experiment we reconstruct an image via its complex quadratic measurements given by \eqref{eq:quad_eqs}. For image reconstruction the iterative WF reconstruction is terminated if the distance \eqref{eq:compute_min_dist} between successive iterations is less than $10^{-6}$ or if 2500 iterations are completed.\footnote{Note that we use distance rather than \emph{relative} distance. This is a stricter termination criteria when $\|\vx\| > 1$.} The measurement matrices are from Section \ref{sec:rot_exp_family} with $q=0$, which corresponds to a complex random Gaussian measurement model. We reconstruct the three color channels of an image of size $n = 22 \times 15 = 330$ pixels separately when $\frac{m}{n}=4$. Fig. \ref{fig:logos} shows the absolute value of the spectral initialization and the corresponding successful reconstruction.

We define the relative error as  $\frac{\||\vz| - \vx\|}{\|\vx\|}$, where $|\vz|$ is the absolute value of the recovered image and $\vx$ is the original image. We further define relative distance as $\frac{\mathrm{dist}\left(\vx, \vz\right)}{\|\vx\|_2}$. When $\frac{m}{n}=4$, the relative error of the spectral initialization is $0.34$. The relative distances between the three channels of the original image and their respective spectral initializations are $0.53$, $0.49$ and $0.51$. The reconstruction relative error is $4.78 \times 10^{-7}$. The relative distances between the three channels of the original and their respective reconstructions are $5.45 \times 10^{-7}$, $7.73 \times 10^{-7}$ and $8.66 \times 10^{-7}$.

We also run our algorithm from the beginning to reconstruct the same image for varying number of measurements. For each value of $\frac{m}{n}$ we draw a new set of measurement matrices, calculate a new spectral initialization and use the drawn measurement matrices and initialization for Wirtinger flow updates. Fig. \ref{fig:logos} shows a failure case when $\frac{m}{n} = 1$ and Fig. \ref{fig:logo_results} shows the relative distances of the recovered images for each channel.

\begin{figure}[tpb]
    \centering
    \includegraphics[width=0.9\columnwidth]{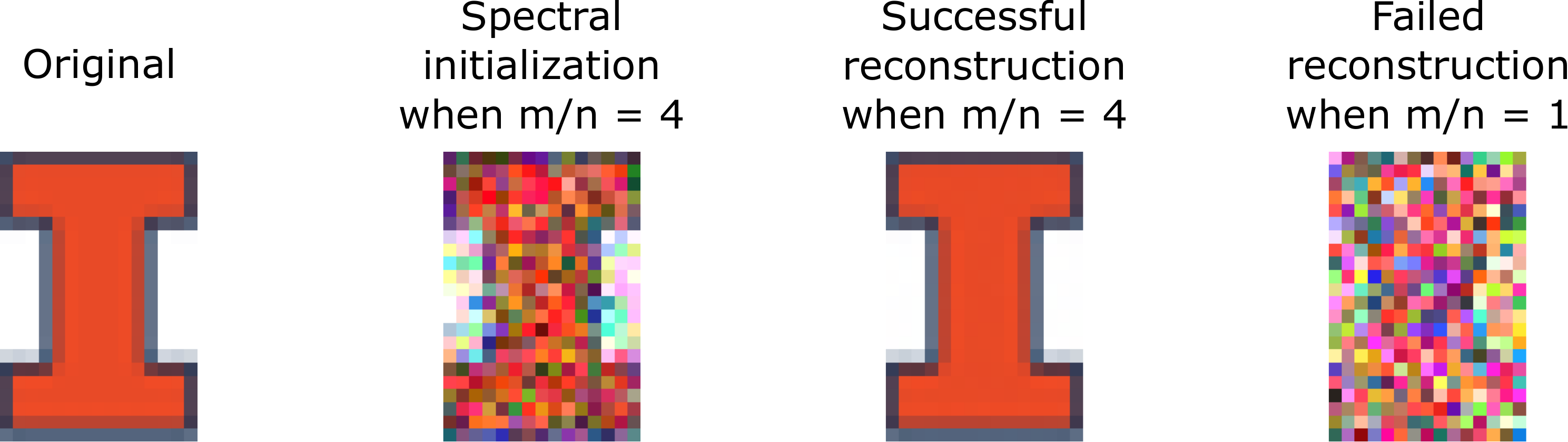}
    \caption{Spectral initialization and the successful reconstruction of the University of Illinois at Urbana-Champaign logo from its complex random quadratic Gaussian measurements when $\frac{m}{n} = 4$. A failed reconstruction when $\frac{m}{n} = 1$ is also shown. The image is of size $n = 22 \times 15 = 330$ pixels.}
    \label{fig:logos}
\end{figure}

\begin{figure}[tpb]
    \centering
    \includegraphics[width=0.9\columnwidth]{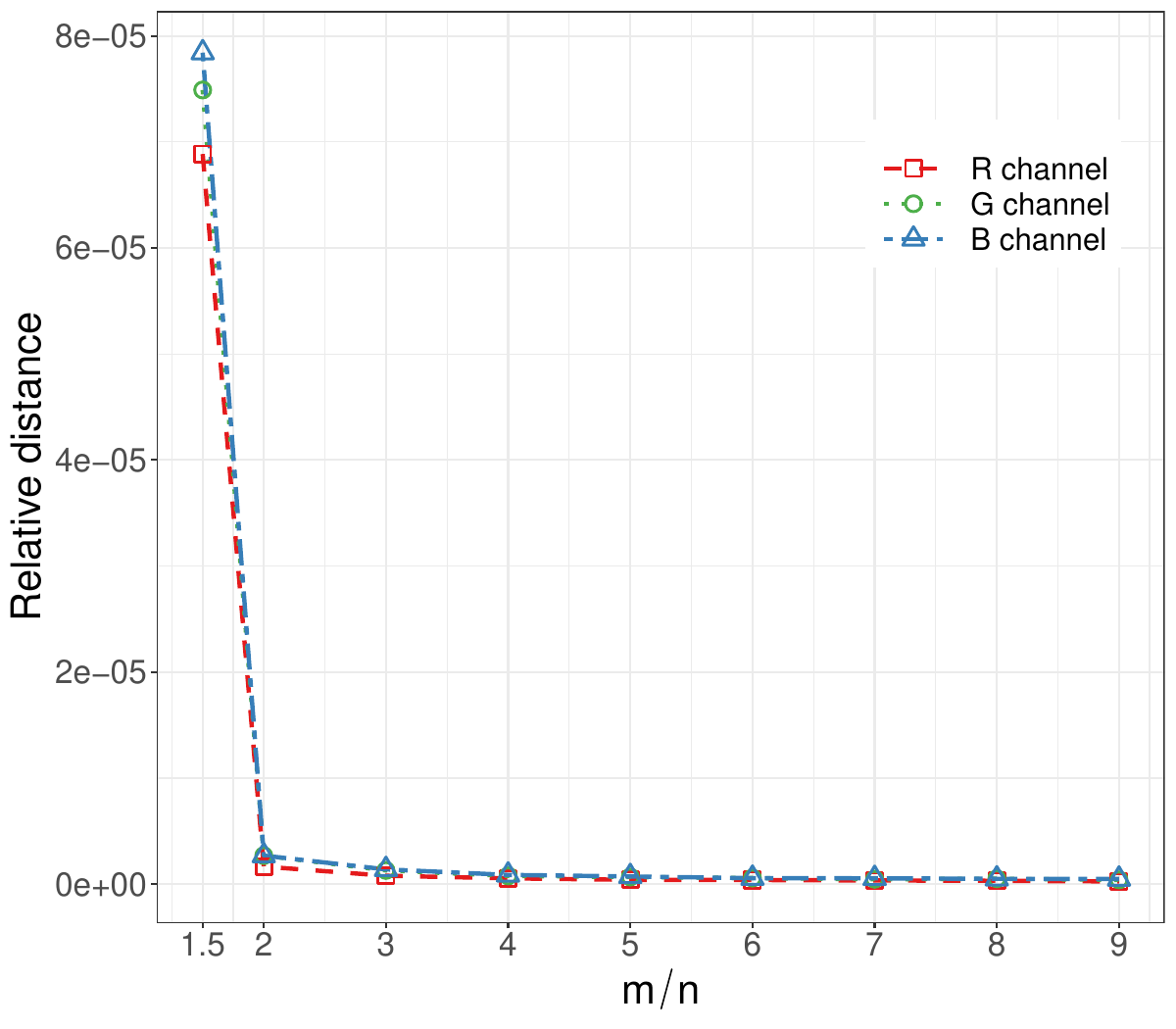}
    \caption{Reconstruction performance of the image from Fig. \ref{fig:logos} with varying number of measurements.}
    \label{fig:logo_results}
\end{figure}

\subsection{Comparison between spectral and random initializations}

\begin{figure*}[tpb]
    \centering
    \subfigure[]{
    \includegraphics[height=0.29\textheight]{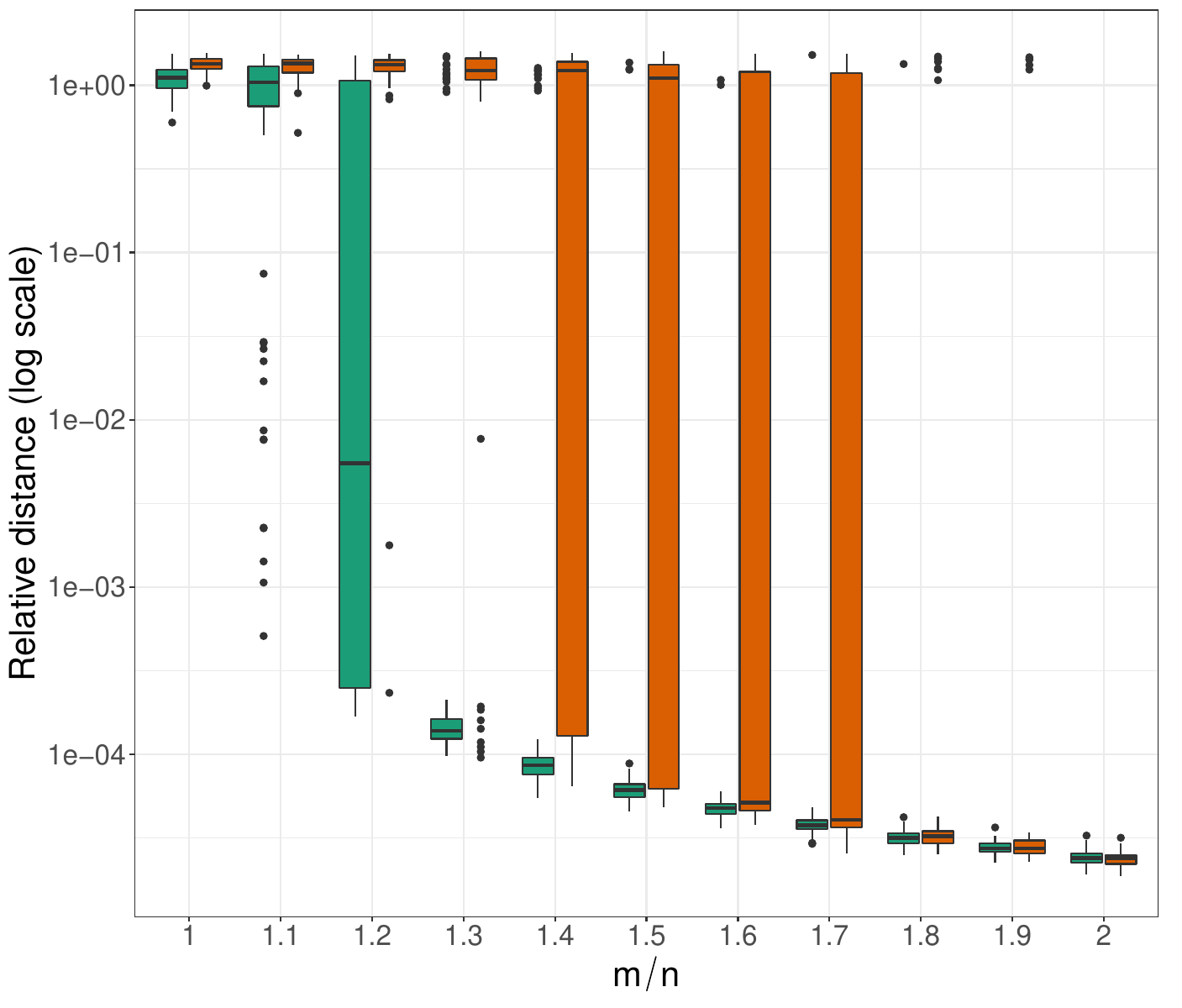}}
    \subfigure[]{
    \includegraphics[height=0.29\textheight]{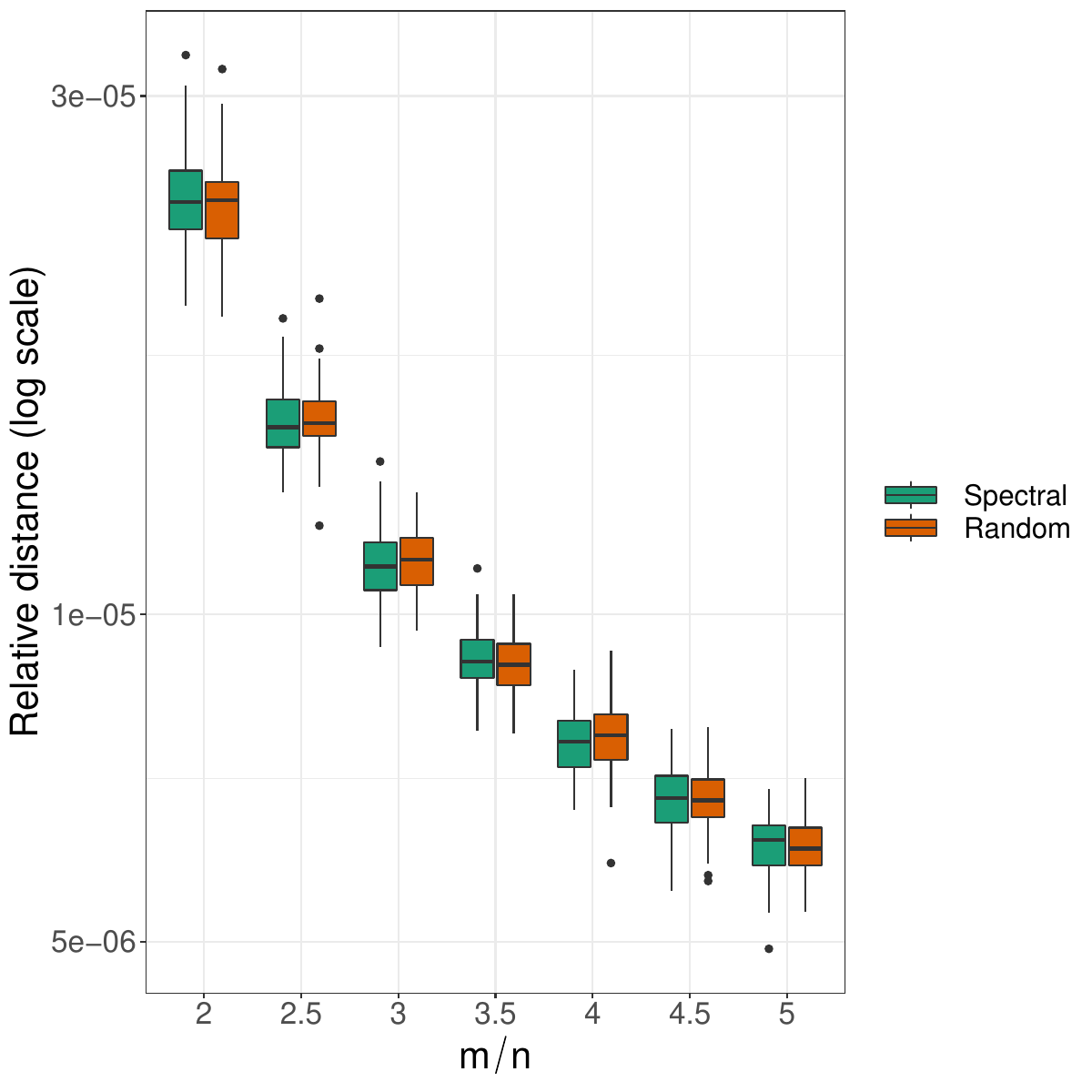}}
    \caption{The relative distances between the recovered signal and the true signal using spectral initialization and random initialization strategies: (a) the sampling rate $\frac{m}{n}\in[1,2]$, (b) the sampling rate $\frac{m}{n}\in[2,5]$.}
    \label{fig:init_comparisons}
\end{figure*}

Random initialization has been shown to be a viable alternative to spectral initialization when solving the phase retrieval problem \cite{Chen2018RandInt}. We next compare the two initialization strategies using the relative distance between the recovered signal and the true signal. Once again we run 100 random trials for each $\frac{m}{n}$ value with $n=100$. We separately analyze the behaviors of the two initializations in the low-oversampling regime, $\frac{m}{n}\in[1,2]$, and the high-oversampling regime, $\frac{m}{n}\in[2,5]$. In each trial we generate a random signal $\vx \in \mathbb{C}^n$ and $m$ standard complex multivariate Gaussian measurement matrices to produce $m$ complex quadratic measurements. We compare our proposed spectral initialization against a random initialization that is generated from a standard complex Gaussian distribution. The spectral initialization is computed using 10 power iterations. Both initializations are scaled so that their norms match the estimated signal norm given by \eqref{eq:unknown_norm}. The same termination criteria is used as in section \ref{subsec:exp_phase_transition}.

The relative distances between the recovered signal and the true signal using different initialization strategies are shown in Fig. \ref{fig:init_comparisons}. We can see that when $1.1\leq\frac{m}{n}\leq1.7$, the spectral initialization with only 10 power iterations performs better than the random initialization. When $\frac{m}{n}\geq1.8$, the two initialization strategies perform almost equally well.

\subsection{Computational efficiency of spectral initialization}
When computing the spectral initialization, the power method is more computationally efficient than a full SVD. In our experiments, we observed that using only 10 power iterations is generally enough to get a good spectral initializer that matches the performance of the exact SVD. In Table \ref{tab:init_comparison} we compare the full SVD and the power method when $\frac{m}{n}=4$ across 5 random trials in terms of the runtimes and the relative distances between the spectral initializers and the true signal. We used up to 24 cores of a system with two 20-core IBM 2.4GHz POWER9 CPUs and up to 115.2GB of RAM. We can see that the power method is preferable since we are only interested in obtaining the leading singular vector.

\begin{table}[tpb]
\renewcommand{\arraystretch}{1.2} 
\centering
\caption{Comparison of SVD and the power method.}
\textbf{}\begin{tabular}{@{}lrrcrr@{}}
\toprule
\multirow{2}{*}{$n$} & \multicolumn{2}{c}{Average time [seconds] } & \phantom{0} & \multicolumn{2}{c}{Average rel. distance} \\ \cmidrule{2-3} \cmidrule{5-6} 
 & SVD & Power && SVD & Power \\ \midrule
2000 & 4.35 & 0.34 && 0.522 & 0.529 \\ 
5000 & 54.83 & 4.35 && 0.518 & 0.524 \\ 
7000 &  115.48 & 7.25 && 0.517 & 0.521 \\ 
10000 &  381.94 & 37.62 && 0.517 & 0.558 \\
\bottomrule
\end{tabular}
\label{tab:init_comparison}
\end{table}

\section{Conclusion}
\label{sec:con_fw}

We addressed the problem of recovering a signal $\vx\in\mathbb{C}^n$ from a system of complex random quadratic equations $y_i=\vx^*\vA_i\vx$, for rotation-invariant sub-Gaussian measurement matrices $\{\vA_i\}_{i=1}^m$. Our analysis complements the existing results on quadratic equations with real measurements and rank-1 positive semidefinite measurement matrices, and extends them to full-rank complex matrices. Since our measurements matrices have uncorrelated entries, the new proofs based on (and including) Lemma \ref{lemma:spectral_norm_concen_all} can be made much simpler than those for phase retrieval, where the entries of measurement matrices are correlated. Our main result has a standard form: we show that when the number of complex measurements exceeds the length of $\vx$ multiplied by some sufficiently large $C$, then with high probability: 1) the spectral initializer $\vz^{(0)}$ is close to a global optimizer; 2) the WF iterates initialized with $\vz^{(0)}$ converge linearly to a global optimizer. Numerical experiments corroborate the theoretical analysis and show that a global optimum can be successfully recovered when $m$ is sufficiently large. 

Recent phase retrieval works showed that a regularized spectral initialization and WF update could improve the robustness and performance of the recovery algorithm \cite{Chen2015Trunc,Wang2018Trunc, wang2017solving}. Chen et al. further proved that vanilla gradient descent with random initialization enjoys favorable convergence guarantees in solving the phase retrieval problem \cite{Chen2018RandInt}. Recent works \cite{Srinadh:2016:GlobalOpt,Park:2017,Ge:2017:Geometric,Li:2018:NonconvexGeo} on the optimization landscape of the low-rank matrix recovery problem give us reason to believe similar optimization landscape could also exist in this case. Our ongoing work involves extending the latest developments in phase retrieval to the aforementioned rotation-invariant sub-Gaussian measurement model. Perhaps more importantly, we hope to adapt our approach to work with quadratic measurements obtained using general high-rank complex matrices that arise in key applications discussed in the introduction.

\appendices

\section{Proofs for rotation-invariant sub-Gaussian measurement model}
\label{app:proofs}

\subsection{Proof of Lemma \ref{lemma:spectral_norm_concen}}
\label{proof:lemma:spectral_norm_concen}
\begin{proof}
Recall that we assume the real and imaginary coefficients of the entries have unit variance, $ \mathbb{E}\big[{A_{i,rc}^{(R)}}^2\big]=\mathbb{E}\big[{A_{i,rc}^{(I)}}^2\big]=1\,.$

For fixed $\vp,\vq$, by rotation invariance, we can choose
\begin{align}
\vp&=\ve_1\\
\vq&=r_1e^{\vj\phi_1}\ve_1+ r_2e^{\vj\phi_2}\ve_2\,,
\end{align}
where $r_1,\ r_2$ are non-negative real numbers satisfying $r_1^2+r_2^2=1$. Let $b_{i}=r_1e^{\vj\phi_1}\overline{A}_{i,11}+r_2e^{\vj\phi_2}\overline{A}_{i,21}$, and $\widetilde{\vA}_i$ denote the matrix $\vA_i$ with the $(1,1)$-th and $(2,1)$-th entries replaced by $0$s. Then
\begin{align}
\label{eq:lemma_2_first}
\begin{split}
&\left\|\frac{1}{m}\sum_{i=1}^m\vp^*\vA_i^*\vq\cdot\vA_i - 2\vq\vp^*\right\|\\
&\leq \left|\frac{1}{m}\sum_{i=1}^mb_iA_{i,11} - 2r_1e^{\vj\phi_1}\right|+ \left|\frac{1}{m}\sum_{i=1}^mb_iA_{i,21} - 2r_2e^{\vj\phi_2}\right|\\
&\quad+\left\|\frac{1}{m}\sum_{i=1}^mb_i\widetilde{\vA}_i\right\|\\
&=|B_1|+|B_2|+\left\|\vH\right\|\\
&\leq |\mathrm{Re}(B_1)|+|\mathrm{Re}(B_2)|+|\mathrm{Im}(B_1)|+|\mathrm{Im}(B_2)|+\|\vH\|\,,
\end{split}
\end{align}
where $\vH=\frac{1}{m}\sum_{i=1}^mb_i\widetilde{\vA}_i$, $B_1=\frac{1}{m}\sum_{i=1}^mb_iA_{i,11} - 2r_1e^{\vj\phi_1}$, $B_2=\frac{1}{m}\sum_{i=1}^mb_iA_{i,21} - 2r_2e^{\vj\phi_2}$, $\mathrm{Re}(B_1)$ denotes the real coefficient of $B_1$ and $\mathrm{Im}(B_1)$ denotes the imaginary coefficient of $B_1$.

For the first term of \eqref{eq:lemma_2_first}, we have:
\begin{align}
\label{eq:spec_con_first_term_ub}
\begin{split}
&\left|\textnormal{Re}\left(B_1\right)\right|\\
&\leq\left|\frac{1}{m}\sum_{i=1}^mr_1\cos\phi_1\left(|A_{i,11}|^2-2\right)\right|+\left|\frac{1}{m}\sum_{i=1}^mF_i\right|\,,\\
\end{split}
\end{align}
where
\begin{align}
\begin{split}
    F_i &= r_2\cos\phi_2\left(A_{i,21}^{(R)}A_{i,11}^{(R)}+A_{i,21}^{(I)}A_{i,11}^{(I)}\right)\\
    &\quad-r_2\sin\phi_2\left(A_{i,21}^{(R)}A_{i,11}^{(I)}-A_{i,21}^{(I)}A_{i,11}^{(R)}\right).
\end{split}
\end{align}
One can verify that $r_1\cos\phi_1\left(|A_{i,11}|^2-2\right)$ is a centered subexponential random variable (cf. \cite[Lemma 5.14]{vershynin_2012}). Similarly, $F_i$ is also a centered subexponential random variable. Using the Bernstein-type inequality \cite[Proposition 5.16]{vershynin_2012}, we have:
\begin{align}
\label{eq:lemma_2_1st_term_ub}
\begin{split}
&\mathrm{Pr}\left(\left|\frac{1}{m}\sum_{i=1}^mr_1\cos\phi_1\left(|A_{i,11}|^2-2\right)\right|\geq\frac{\nu}{12}\right) \\
&\quad \leq 2\exp\left(-m\cdot\min\left\{\frac{c\nu^2}{144K_1^2},\frac{c\nu}{12K_1}\right\}\right)
\end{split}\\
\label{eq:lemma_2_1st_term_ub_2}
\begin{split}
&\mathrm{Pr}\left(\left|\frac{1}{m}\sum_{i=1}^mF_i\right|\geq\frac{\nu}{12}\right)\\
&\quad\leq 2\exp\left(-m\cdot\min\left\{\frac{c\nu^2}{144K_2^2},\frac{c\nu}{12K_2}\right\}\right)\,,
\end{split}
\end{align}
where $c>0$ is some absolute constant and $K_1$ and $K_2$ are the respective subexponential norms. Combining \eqref{eq:spec_con_first_term_ub}, \eqref{eq:lemma_2_1st_term_ub} and \eqref{eq:lemma_2_1st_term_ub_2}, we then have
\begin{align}
\label{eq:lemma_2_bound_one}
\begin{split}
\mathrm{Pr}\left(\left|\textnormal{Re}\left(B_1\right)\right| < \frac{\nu}{6}\right) &\geq 1-4\exp\left(-m\cdot\widehat{C}_1(\nu)\right)\,,
\end{split}
\end{align}
where $\widehat{C}_1(\nu)$ is a constant depending on $\nu$,
\begin{align}
\widehat{C}_1(\nu) = \min\left\{\frac{c\nu^2}{144K_1^2},\frac{c\nu}{12K_1},\frac{c\nu^2}{144K_2^2},\frac{c\nu}{12K_2}\right\}\,.
\end{align}

For the second, third and fourth term of \eqref{eq:lemma_2_first}, we obtain similarly:

\begin{align}
\label{eq:lemma_2_bound_two}
\mathrm{Pr}\left(\left|\textnormal{Re}\left(B_2\right)\right| < \frac{\nu}{6}\right) & \geq 1-4\exp\left(-m\cdot\widehat{C}_2(\nu)\right) \\ 
\label{eq:lemma_2_bound_three}
\mathrm{Pr}\left(\left|\textnormal{Im}\left(B_1\right)\right| < \frac{\nu}{6}\right) & \geq 1-4\exp\left(-m\cdot\widehat{C}_3(\nu)\right) \\
\label{eq:lemma_2_bound_four}
\mathrm{Pr}\left(\left|\textnormal{Im}\left(B_2\right)\right| < \frac{\nu}{6}\right) &\geq1-4\exp\left(-m\cdot\widehat{C}_4(\nu)\right)\,,
\end{align}
where $\widehat{C}_2(\nu),\widehat{C}_3(\nu),\widehat{C}_4(\nu)$ are some constants depending on $\nu$.

To compute an upper bound on the spectral norm $\| \vH \|$ in \eqref{eq:lemma_2_first}, we adapt an approach from \cite[Theorem 5.39]{vershynin_2012}. The idea is to bound $\left|\vu^*\vH\bv\right|$ uniformly for all $\vu,\bv\in\mathbb{C}^n$ on the unit sphere $\mathscr{S}^{n-1}$. In order to take the union bound over $\vu$s and $\bv$s, the unit sphere $\mathscr{S}^{n-1}$ is first discretized using an $\epsilon$-net $\mathscr{N}_\epsilon$ \cite[Definition 5.1]{vershynin_2012}, for $\epsilon \in [0, 1)$. For every fixed pair $(\vu,\bv)$, we establish a high-probability upper bound on $\left|\vu^*\vH\bv\right|$, and then take the union bound over  $(\vu,\bv) \in \mathscr{N}_\epsilon\times\mathscr{N}_\epsilon$, taking care of the adjustments so that the result holds over $\mathscr{S}^{n-1} \times \mathscr{S}^{n-1}$.

\paragraph{Approximation} We first bound the error of approximating $\| \vH \|$ using a $(\vu, \bv)$ from $\mathscr{N}_\epsilon \times \mathscr{N}_\epsilon$. Suppose $\vu_1,\bv_1\in\mathscr{S}^{n-1}$ is chosen such that $\|\vH\|=|\langle \vH\bv_1,\vu_1\rangle|$, and choose $\vu_2,\bv_2\in\mathscr{N}_\epsilon$ that approximate $\vu_1,\bv_1$ as $\|\vu_1-\vu_2\|_2\leq\epsilon$, $\|\bv_1-\bv_2\|_2\leq\epsilon$. We get
\begin{align}
\begin{split}
&|\langle\vH\bv_1,\vu_1\rangle-\langle\vH\bv_2,\vu_2\rangle|\\
&=|\vu_1^*\vH(\bv_1-\bv_2)+(\vu_1^*-\vu_2^*)\vH\bv_2|\\
&\leq\|\vH\|\|\vu_1\|_2\|\bv_1-\bv_2\|_2+\|\vH\|\|\vu_1-\vu_2\|_2\|\bv_2\|_2\\
&\leq 2\epsilon\|\vH\|\,.
\end{split}
\end{align}
It follows that
\begin{align}
|\langle\vH\bv_2,\vu_2\rangle|\geq(1-2\epsilon)\cdot\|\vH\|\,.
\end{align}
Taking the maximum over all $\vu_2,\bv_2$ in the above inequality, we obtain the bound
\begin{align}
\label{eq:spec_norm_lemma}
\|\vH\|\leq (1-2\epsilon)^{-1}\cdot\max_{(\vu_2,\bv_2)\in\mathscr{N}_\epsilon\times\mathscr{N}_\epsilon}|\langle\vH\bv_2,\vu_2\rangle|\,.
\end{align}
We now choose $\epsilon=\frac{1}{4}$. According to \cite[Lemma 5.2]{vershynin_2012}, there exists a $\frac{1}{4}$-net with cardinality $\left|\mathscr{N}_{\sfrac{1}{4}}\right|\leq 9^n$. Since we are maximizing over $(\vu,\bv)\in\mathscr{N}_{\sfrac{1}{4}}\times\mathscr{N}_{\sfrac{1}{4}}$, the total cardinality is bounded as $\left|\mathscr{N}_{\sfrac{1}{4}}\right|^2\leq 81^n$ so that
\begin{align}
\label{eq:bound_spectral_norm_net}
\|\vH\|\leq \max_{(\vu,\bv)\in\mathscr{N}_{\sfrac{1}{4}}\times\mathscr{N}_{\sfrac{1}{4}}}2|\vu^*\vH\bv|\,.
\end{align}

\paragraph{Concentration} 
For a fixed $(\vu,\bv) \in \mathscr{N}_{\sfrac{1}{4}}\times\mathscr{N}_{\sfrac{1}{4}}$, we have:
\begin{align}
\label{eq:lemma_2_second_ub}
\begin{split}
\left|\frac{1}{m}\sum_{i=1}^mb_i\vu^*\widetilde{\vA}_i\bv\right|&=\left|\frac{1}{m}\sum_{i=1}^mb_i\cdot G_i\right|\\
&\leq\left|\frac{1}{m}\sum_{i=1}^m\left(b_i^{(R)}G_i^{(R)}-b_i^{(I)}G_i^{(I)}\right)\right|\\
&\quad+\left|\frac{1}{m}\sum_{i=1}^m\left(b_i^{(R)}G_i^{(I)}+b_i^{(I)}G_i^{(R)}\right)\right|\,,
\end{split}
\end{align}
where $G_i=\sum_{kl}\overline{u}_kv_l\cdot\widetilde{\vA}_{i,kl}$. For the first term of \eqref{eq:lemma_2_second_ub}, $b_i^{(R)}$ and $b_i^{(I)}$ are linear combinations of the real and imaginary coefficients of $A_{i,11}$ and $A_{i,21}$. On the other hand, $G_i^{(R)}$ and $G_i^{(I)}$ are linear combinations of the coefficients of the entries in $\widetilde{\vA}_i$ that do not contain $A_{i,11}$ and $A_{i,21}$. We can check that $b_i^{(R)}G_i^{(R)}-b_i^{(I)}G_i^{(I)}$ is a centered subexponential random variable as before. Using the Bernstein-type inequality\cite[Proposition 5.16]{vershynin_2012}, we have:
\begin{align}
\label{eq:77_u1}
\begin{split}
&\mathrm{Pr}\left(\left|\frac{1}{m}\sum_{i=1}^mb_i^{(R)}G_i^{(R)}-b_i^{(I)}G_i^{(I)}\right|\geq\frac{\nu}{12}\right)\\
&\leq 2\exp\left(-m\cdot\min\left\{\frac{c\nu^2}{144K_{3}^2},\frac{c\nu}{12K_{3}}\right\}\right)\,,
\end{split}
\end{align}
where $c>0$ is some absolute constant and $K_{3}$ is the subexponential norm of $b_i^{(R)}G_i^{(R)}-b_i^{(I)}G_i^{(I)}$. We get a similar result for the second term of \eqref{eq:lemma_2_second_ub},
\begin{align}
\label{eq:77_u2}
\begin{split}
    &\mathrm{Pr}\left(\left|\frac{1}{m}\sum_{i=1}^mb_i^{(R)}G_i^{(I)}+b_i^{(I)}G_i^{(R)}\right|\geq\frac{\nu}{12}\right)\\
    &\leq 2\exp\left(-m\cdot\min\left\{\frac{c\nu^2}{144K_{4}^2},\frac{c\nu}{12K_{4}}\right\}\right)\,,
\end{split}
\end{align}
where $c>0$ is some absolute constant and $K_{4}$ is the subexponential norm of $b_i^{(R)}G_i^{(I)}+b_i^{(I)}G_i^{(R)}$.

Combining \eqref{eq:lemma_2_second_ub},\eqref{eq:77_u1},\eqref{eq:77_u2}, we have
\begin{align}
\begin{split}
&\mathrm{Pr}\left(2\left|\frac{1}{m}\sum_{i=1}^mb_i\cdot G_i\right| < \frac{\nu}{3}\right) \geq 1-4\exp\left(-m\cdot\widehat{C}_5(\nu)\right),
\end{split}
\end{align}
where $\widehat{C}_5(\nu)$ is a constant depending on $\nu$ and the measurement model,
\begin{align}
    \widehat{C}_5(\nu)=\min\left\{ \frac{c\nu^2}{144K_{3}^2},\ \frac{c\nu}{12K_{3}},\ \frac{c\nu^2}{144K_{4}^2},\ \frac{c\nu}{12K_{4}} \right\}\,.
\end{align}

\paragraph{Union bound} Taking the union bound over all unit vectors $(\vu,\bv) \in\mathscr{N}_{\sfrac{1}{4}}\times\mathscr{N}_{\sfrac{1}{4}}$ with cardinality $\left|\mathscr{N}_{\sfrac{1}{4}}\right|^2\leq 81^n$, if $m\geq Cn$,
\begin{align}
\begin{split}
&\mathrm{Pr}\left(\max_{\vu,\bv\in\mathscr{N}_{\sfrac{1}{4}}\times\mathscr{N}_{\sfrac{1}{4}}} 2\left|\frac{1}{m}\sum_{i=1}^mb_i\vu^*\widetilde{\vA}_i\bv\right| \geq \frac{\nu}{3}\right)\\
&\leq 81^n\cdot 4\exp\left(-m\cdot\widehat{C}_5(\nu)\right)\\
&\leq 4\exp\left(-m\cdot\left(\widehat{C}_5(\nu)-C^{-1}\ln 81\right)\right)\,.
\end{split}
\end{align} 
Using \eqref{eq:spec_norm_lemma}, we have
\begin{align}
\label{eq:lemma_2_bound_five}
\begin{split}
&\mathrm{Pr}\left(\left\|\frac{1}{m}\sum_{i=1}^mb_i\widetilde{\vA}_i\right\| < \frac{\nu}{3}\right)\\
&\geq \mathrm{Pr}\left(\max_{\vu,\bv\in\mathscr{N}_{\sfrac{1}{4}}\times\mathscr{N}_{\sfrac{1}{4}}} 2\left|\frac{1}{m}\sum_{i=1}^mb_i\vu^*\widetilde{\vA}_i\bv\right| < \frac{\nu}{3}\right)\\
&\geq 1-4\exp\left(-m\cdot\left(\widehat{C}_5(\nu)-C^{-1}\ln 81\right)\right)\,.
\end{split}
\end{align}
If $m\geq Cn$, using \eqref{eq:lemma_2_first} and combining all the bounds so far \eqref{eq:lemma_2_bound_one}, \eqref{eq:lemma_2_bound_two}, \eqref{eq:lemma_2_bound_three}, \eqref{eq:lemma_2_bound_four}, \eqref{eq:lemma_2_bound_five}, we get
\begin{align}
\label{eq:finalbound-lemma1}
\begin{split}
&\mathrm{Pr}\left(\left\|\frac{1}{m}\sum_{i=1}^m\vp^*\vA_i^*\vq\cdot\vA_i - 2\vq\vp^*\right\| < \nu\right)\\
& \geq 1-20\exp\big(-m\cdot C_1(C,\nu)\big)\,,
\end{split}
\end{align}
where $C_1(C,\nu)$ depends on $C$ and $\nu$, but not on $m$.
\begin{align}
\begin{split}
C_1(C,\nu)= &\min\left\{\widehat{C}_1(\nu),\ \widehat{C}_2(\nu),\ \widehat{C}_3(\nu),\ \widehat{C}_4(\nu),\right.\\
&\quad\quad\quad\quad\quad\quad\quad\left.\ \widehat{C}_5(\nu)-C^{-1}\ln81\right\}\,.
\end{split}
\end{align}
When $C$ is sufficiently large, we have that $C_1(C,\nu)>0$ so that both sides of \eqref{eq:finalbound-lemma1} go to $1$ as $m \to \infty$.

\end{proof}

\subsection{Proof of Lemma \ref{lemma:initialization}}
\label{proof:lemma:initialization}

\begin{proof}
Let $\{\vu_0,\bv_0\}$ be the leading left and right singular vectors of $\vS$, and $\tau_0$ be its largest singular value. Using Lemma \ref{lemma:spectral_norm_concen}, the following holds with probability at least $1-20\exp\big(-m\cdot C_1(C,\delta)\big)$:
\begin{align}
\label{eq:proof-lemma-ineq1}
\begin{split}
\left|\tau_0-\vu_0^* (2\vx\vx^*) \bv_0\right|&=\left|\vu_0^*\left(\vS-2\vx\vx^*\right)\bv_0\right|\\
&\leq\|\vS-2\vx\vx^*\|\\
&\leq\delta\|\vx\|_2^2\,.
\end{split}
\end{align}
Hence, on this event, $\vu_0^* (2\vx\vx^*) \bv_0\geq\tau_0-\delta\|\vx\|_2^2$ and we have with at least the same probability that
\begin{align}
\begin{split}
\tau_0&\geq\frac{1}{\|\vx\|_2^2}\vx^*\vS\vx\\
&=\frac{1}{\|\vx\|_2^2}\vx^*\left(\vS-2\vx\vx^*\right)\vx+2\|\vx\|_2^2\\
&\geq (2-\delta)\|\vx\|_2^2\,.
\end{split}
\end{align}

In the following proof the spectral initializer $\vz^{(0)}$ is constructed from $\bv_0$. When $\vz^{(0)}$ is constructed from $\vu_0$, the proof is similar.
\begin{enumerate}
\item When the signal norm $\| \vx \|_2$ is known, we can choose $\vz^{(0)}=\|\vx\|_2\cdot\bv_0$ as the spectral initializer. Since $|\vu_0^*\vx|\leq \|\vx\|_2$, we have 
\begin{align}
\label{eq:lemma_2_first_ineq_known}
\begin{split}
2\left|\vx^*\bv_0\cdot\|\vx\|_2\right|&\geq 2|\vu_0^*\vx|\cdot|\vx^*\bv_0|\\
&\geq \vu_0^* (2\vx\vx^*) \bv_0\\
&\geq\tau_0-\delta\|\vx\|_2^2\\
&\geq (2-2\delta)\|\vx\|_2^2\,.
\end{split}
\end{align}
Using \eqref{eq:compute_min_dist}, the squared distance between the spectral initializer $\vz^{(0)}=\|\vx\|_2\bv_0$ and $\vx$ is then bounded as
\begin{align}
\begin{split}
\mathrm{dist}^2\left(\|\vx\|_2\bv_0,\vx\right)
&=\min_{\phi\in(0,2\pi]}\left\|\|\vx\|_2\bv_0 -\vx e^{\vj\phi}\right\|_2^2\\
&=\|\vx\|_2^2\|\bv_0\|_2^2+\|\vx\|_2^2-2|\vx^*\bv_0\cdot\|\vx\|_2|\\
&\leq 2\|\vx\|_2^2-(2-2\delta)\|\vx\|_2^2 \\
&= 2\delta\|\vx\|_2^2\leq\frac{51}{24}\delta\|\vx\|_2^2\,,
\end{split}
\end{align}
with probability at least $1-20\exp\big(-m\cdot C_1(C,\delta)\big)$.

\item When the signal norm $\|\vx\|_2$ is unknown, we estimate it from $R=\frac{1}{2m}\sum_{i=1}^m\overline{y}_iy_i$. By rotation invariance of the sub-Gaussian matrix $\vA_i$, we can simply assume $\vx=\|\vx\|_2\ve_1$ so that
\begin{align}
R=\frac{\|\vx\|_2^4}{2m}\sum_{i=1}^m|A_{i,11}|^2\,.
\end{align}
Using \eqref{eq:lemma_2_1st_term_ub} in the proof of Lemma \ref{lemma:spectral_norm_concen} ($\phi_1=0$ in this case), we know that on the same event on which \eqref{eq:proof-lemma-ineq1} holds (of probability $\geq 1-20\exp\big(-m\cdot C_1(C,\delta)\big)$), it also holds that
\begin{align}
\label{eq:est_signal_norm}
\left(2-\frac{\delta}{12}\right)m\leq\sum_{i=1}^m|A_{i,11}|^2\leq\left(2+\frac{\delta}{12}\right)m\,.
\end{align}
We thus have
\begin{align}
\left(1-\frac{\delta}{24}\right)\|\vx\|_2^4\leq R\leq\left(1+\frac{\delta}{24}\right)\|\vx\|_2^4\,,
\end{align}
and choose the spectral initializer as $\vz^{(0)}=\sqrt[4]{R}\bv_0$. Assuming $\delta\in(0,24)$ and using \eqref{eq:lemma_2_first_ineq_known}, we have that
\begin{align}
\label{eq:bd_signal_unknown}
\begin{split}
&\mathrm{dist}^2\left(\sqrt[4]{R}\bv_0,\vx\right)\\
&\leq\sqrt[2]{R}+\|\vx\|_2^2-(2-2\delta)\sqrt[4]{R}\|\vx\|_2\\
&\leq\sqrt[2]{1+\frac{\delta}{24}}\|\vx\|_2^2+\|\vx\|_2^2-2(1-\delta)\sqrt[4]{1-\frac{\delta}{24}}\|\vx\|_2^2\\
&\leq\left(1+\frac{\delta}{24}\right)\|\vx\|_2^2+\|\vx\|_2^2-2(1-\delta)\left(1-\frac{\delta}{24}\right)\|\vx\|_2^2\\
&\leq \frac{51}{24}\delta\|\vx\|_2^2\,,
\end{split}
\end{align}
with probability at least $1-20\exp\big(-m\cdot C_1(C,\delta)\big)$.
\end{enumerate}
\end{proof}

\subsection{Proof of Lemma \ref{lemma:spectral_norm_concen_all}}
\label{proof:lemma:spectral_norm_concen_all}
\begin{proof}
Let $\vG(\vp,\vq) := \frac{1}{m}\sum_{i=1}^m\vp^*\vA_i^*\vq\cdot\vA_i-2\vq\vp^*$. We prove that the bound on the spectral norm $\|\vG(\vp,\vq)\|$ in \eqref{eq:spectral_norm_concen_all} holds with high probability \emph{for all} unit vectors $\vp,\vq$ by combining Lemma \ref{lemma:spectral_norm_concen} with yet another union bound. 
Let $\vp_1,\vq_1\in\mathscr{S}^{n-1}$ such that
\begin{align}
    \|\vG(\vp_1,\vq_1)\|=\max_{(\vp,\vq) \in \mathscr{S}^{n-1}\times\mathscr{S}^{n-1}}\|\vG(\vp,\vq)\|\,.
\end{align}
Let $\vp_2,\vq_2\in\mathscr{N}_\epsilon$ further obey $\vp_1,\vq_1$ as $\|\vp_1-\vp_2\|_2\leq\epsilon$, $\|\vq_1-\vq_2\|_2\leq\epsilon$ (they exist by the definition of an $\epsilon$-net). We can write
\begin{align}
\begin{split}
    &\|\vG(\vp_1,\vq_1)-\vG(\vp_2,\vq_2)\| \\ &=\|\vG(\vp_1,\vq_1-\vq_2)+\vG(\vp_1-\vp_2,\vq_2)\|\\
    &\leq\|\vG(\vp_1,\vq_1-\vq_2)\|+\|\vG(\vp_1-\vp_2,\vq_2)\|\\
    &=\|\vq_1-\vq_2\|_2\cdot\left\|\vG\left(\vp_1,\frac{\vq_1-\vq_2}{\|\vq_1-\vq_2\|_2}\right)\right\|\\
    &\quad+\|\vp_1-\vp_2\|_2\cdot\left\|\vG\left(\frac{\vp_1-\vp_2}{\|\vp_1-\vp_2\|_2},\vq_2\right)\right\|\\
    &\leq\big(\|\vq_1-\vq_2\|_2+\|\vp_1-\vp_2\|_2\big)\cdot\|\vG(\vp_1,\vq_1)\|\\
    &\leq2\epsilon\|\vG(\vp_1,\vq_1)\|\,,
\end{split}
\end{align}
so that
\begin{align}
    \|\vG(\vp_1,\vq_1)\|\leq(1-2\epsilon)^{-1}\|\vG(\vp_2,\vq_2)\|\,.
\end{align}
Taking the maximum over all $\vp_2,\vq_2$ in the above inequality, we get
\begin{align}
\|\vG(\vp_1,\vq_1)\| \leq (1-2\epsilon)^{-1} \max_{\vp_2,\vq_2\in\mathscr{N}_\epsilon\times\mathscr{N}_\epsilon}\|\vG(\vp_2,\vq_2)\|\,.
\end{align}
We again choose $\epsilon=\frac{1}{4}$ so that as in Lemma \ref{lemma:spectral_norm_concen}, $|\mathscr{N}_{\sfrac{1}{4}}\times\mathscr{N}_{\sfrac{1}{4}}| = |\mathscr{N}_{\sfrac{1}{4}}|^2 \leq 81^n$ and
\begin{align}
\label{eq:spec_lemma_all_unit_vec}
    \|\vG(\vp_1,\vq_1)\| \leq \max_{\vp_2,\vq_2\in\mathscr{N}_{\sfrac{1}{4}}\times\mathscr{N}_{\sfrac{1}{4}}} 2\|\vG(\vp_2,\vq_2)\|.
\end{align}

By Lemma \ref{lemma:spectral_norm_concen}, if $m>Cn$ for some sufficiently large $C$,
\begin{align}
    \mathrm{Pr}\left(\left\|\vG(\vp_2,\vq_2)\right\|\geq\frac{\nu}{2}\right)\leq 20\exp\left(-m\cdot C_1\left(C,\frac{\nu}{2}\right)\right)\,,
\end{align}
for fixed unit vectors $\vp_2,\vq_2$. To get a result which holds for all $\vp$ and $\vq$, we take the union bound over $\mathscr{N}_{\sfrac{1}{4}}\times\mathscr{N}_{\sfrac{1}{4}}$. For $m\geq Cn$,
\begin{align}
\begin{split}
    &\mathrm{Pr}\left(\max_{\vp_2,\vq_2\in\mathscr{N}_\epsilon\times\mathscr{N}_\epsilon} 2\|\vG(\vp_2,\vq_2)\|\geq \nu\right)\\
    &\leq 81^n\cdot 20\exp\left(-m\cdot C_1\left(C,\frac{\nu}{2}\right)\right)\\
    &\leq 20\exp\left(-m\cdot\left[C_1\left(C,\frac{\nu}{2}\right)-C^{-1}\ln81\right]\right)\,.
\end{split}
\end{align}
Using \eqref{eq:spec_lemma_all_unit_vec}, we have
\begin{align}
\begin{split}
    &\mathrm{Pr}\left( \text{\Large $\forall$} \vp, \vq \in \mathscr{S}^{n-1}, \|\vG(\vp, \vq)\| < \nu \right)\\
    &= \mathrm{Pr}\left(\|\vG(\vp_1,\vq_1)\|<\nu\right) \\ &\geq \mathrm{Pr}\left(\max_{(\vp_2,\vq_2)\in\mathscr{N}_\epsilon\times\mathscr{N}_\epsilon} 2\|\vG(\vp_2,\vq_2)\|< \nu\right)\\
    &\geq 1-20\exp\big(-m\cdot C_2(C,\nu)\big)\,,
\end{split}    
\end{align}
where $C_2(C,\nu) := C_1\left(C,\frac{\nu}{2}\right)-C^{-1}\ln81 > 0$ for $C$ sufficiently large.

\end{proof}

\subsection{Example of rotation-invariant distributions}
\label{app:ri_exp_fam}

For completeness, we exhibit here one family of rotation-invariant sub-Gaussian distributions. Consider the random variable $\vu=\|\vs\|_2^q\cdot\vs$, where $q\in(0,1)$, and $\vs\in\mathbb{R}^d\sim\mathcal{N}(\vzero,\vI)$, $\vs\neq\vzero$. We have:
\begin{align}
    \vs=\|\vu\|_2^{-\frac{q}{q+1}}\cdot\vu\,.
\end{align}

The entries of the Jacobian matrix $\vJ = d \vs / d \vu$ are given as
\begin{align}
    &\frac{\partial s_i}{\partial u_i}=\|\vu\|_2^{-\frac{q}{q+1}}\left(1-\frac{q}{q+1}\|\vu\|_2^{-2}\cdot u_i^2\right)\\
    &\frac{\partial s_i}{\partial u_j}=\|\vu\|_2^{-\frac{q}{q+1}}\left(-\frac{q}{q+1}\|\vu\|_2^{-2}\cdot u_iu_j\right),\quad i\neq j\,.
\end{align}
The Jacobian matrix $\vJ$ is thus
\begin{align}
    \vJ=\|\vu\|_2^{-\frac{q}{q+1}}\left(\vI-\frac{q}{q+1}\|\vu\|_2^{-2}\cdot\vu\vu^\mathrm{T}\right)\,,
\end{align}
with the determinant given by
\begin{align}
\begin{split}
    \mathrm{det}(\vJ)&=\|\vu\|_2^{-\frac{qd}{q+1}}\mathrm{det}\left(\vI-\frac{q}{q+1}\|\vu\|_2^{-2}\cdot\vu\vu^\mathrm{T}\right)\\
    &=\|\vu\|_2^{-\frac{qd}{q+1}}\left(1-\frac{q}{q+1}\|\vu\|_2^{-2}\cdot\vu^\mathrm{T}\vu\right)\\
    &=\frac{1}{q+1}\|\vu\|_2^{-\frac{qd}{q+1}}\,.
\end{split}
\end{align}
We obtain the expression for the pdf as
\begin{align}
    p(\vu)&=\frac{1}{q+1}\|\vu\|_2^{-\frac{qd}{q+1}}(2\pi)^{-\frac{d}{2}}\exp\left(-\frac{1}{2}\|\vu\|_2^{\frac{2}{q+1}}\right)\,.
\end{align}
We can see that $p(\vu)$ only depends on the norm $\|\vu\|_2$, and is thus invariant under unitary transform.  For $q = -1$ we obtain a uniform distribution on the sphere, while $q = 0$ gives the Gaussian distribution. In all cases, when $q \in [-1, 0]$, it is easy to check that the moments of the random variables generated as above are suitably bounded (the tails decay faster than the Gaussian) so that these random variables are sub-Gaussian.

\bibliographystyle{IEEEbib}
\bibliography{refs}

\end{document}